\documentclass[12pt]{article}
\usepackage{amscd,amsmath,amssymb,amstext,amsthm,exscale,latexsym}
\usepackage{graphicx}
\newcommand {\al}   {\alpha}       \newcommand {\bt}  {\beta}
\newcommand {\g }   {\gamma}       
\newcommand {\dl}   {\delta}       \newcommand {\e }  {\epsilon}

         \newcommand {\ph}  {\phi}
\newcommand {\vf }  {\varphi}      
         \newcommand {\om}  {\omega}

\newcommand {\pl}   {\partial}     
\newcommand   {\const}{{\sf\,const}}     \newcommand   {\diag}{{\sf\,diag\,}}
\newcommand {\MM}  {{\mathbb M}}   
\newcommand {\MO}  {{\mathbb O}}   
   \newcommand {\MR}  {{\mathbb R}}
\newcommand {\MS}  {{\mathbb S}}   \newcommand {\MT}  {{\mathbb T}}
   \newcommand {\MV}  {{\mathbb V}}
\newcommand {\Sa}  {{\textsc{a}}}   \newcommand {\Sb}  {{\textsc{b}}}
\newcommand {\Sm}  {{\textsc{m}}}   \newcommand {\Sn}  {{\textsc{n}}}
   \newcommand {\Sp}  {{\textsc{p}}}
\newcommand {\nequiv} {\equiv\!\!\!\!\!\!/\:}
\newcommand   {\ns}{{\sf\, ns\,}}
            
\newtheorem{lemma}{Lemma}[section]
\newtheorem{prop}{Proposition}[section]
\newtheorem{exa}{Example}[section]
\newtheorem{theorem}{Theorem}[section]
\theoremstyle{definition}
\newtheorem*{cor}{Corollary}
\newtheorem*{defn}{Definition}
\begin{document}
\title     {Complete separation of variables in the geodesic Hamilton--Jacobi
            equation}
\author    {M. O. Katanaev
            \thanks{E-mail: katanaev@mi-ras.ru}\\ \\
            \sl Steklov Mathematical Institute,\\
            \sl 119991, Moscow, ul. Gubkina, 8}
\maketitle
\begin{abstract}
We consider a (pseudo)Riemannian manifold of arbitrary dimension. The
Hamil\-ton--Jacobi equation for geodesic Hamiltonian admits complete separation
of variables for some (separable) metrics in some (separable) coordinate
systems. Separable metrics are very important in mathematics and physics. The
St\"ackel problem is: ``Which metrics admit complete separation of variables in
the geodesic Hamilton--Jacobi equation?'' This problem was solved for inverse
metrics with nonzero diagonal elements, in particular, for positive definite
Riemannian metrics, long ago. However the question is open for indefinite
inverse metrics having zeroes on diagonals. We propose the solution. Separable
metrics are divided into equivalence classes characterised by the number of
Killing vector fields, quadratic indecomposable conservation laws for geodesics,
and the number of coisotropic coordinates. The paper contains detailed proofs,
sometimes new, of previous results as well as new cases. As an example, we list
all canonical separable metrics in each equivalence class in two, three, and
four dimensions. Thus the St\"ackel problem is completely solved for metrics of
any signature in any number of dimensions.
\end{abstract}
\section{Introduction}
Many important and interesting problems in mathematical physics are related to
analysis of geodesics on a (pseudo)Riemannian manifold. In its turn, integration
of geodesic equa\-tions often reduces to solution of the corresponding
Hamilton--Jacobi equation. Models admitting complete separation of variables in
the Hamilton--Jacobi equations are of parti\-cular interest because, in this
case, there are $n$ independent conservation laws in envolution for $n$
dimensional Hamiltonian system, and geodesic equations are integrated in
quadratures. Therefore finding metrics which admit complete separation of
variables is of great importance. This problem is interesting both for
Riemannian (positive definite) metrics and metrics of Lorentz signature, the
latter case being important in gravity models.

In 1891, St\"ackel raised the question: ``Which metrics admit complete
separation of variables in the respective Hamilton--Jacobi equation?''
\cite{Stacke91,Stacke93A,Stacke93B,Stacke93c}, and gave the answer in the case
of diagonal (orthogonal) metrics in the presence of only quadratic
indecomposable conservation laws. These metrics are called St\"ackel or
separable.
The problem attracted much interest of mathematicians and physicists and
became classical. It is clear that if the metric has enough symmetry then it may
admit $n$ envolutive conservation laws. At the same time some St\"ackel metrics
admit complete separation of variables even without any symmetry.

Many interesting and important results for orthogonal separable metrics were
obtained in papers \cite{LeviCi04,Eisenh34,Benent16,BoKoMa22}, but we focus our
attention on nondiagonal metrics.

The separating action functions were found for nondiagonal metrics of arbitrary
signature under the assumption that all diagonal elements differ from zero in
\cite{Acqua08,Burgat11,Acqua12}. This is always true for Riemannian positive
definite metrics. The corresponding separable metrics were derived in
\cite{Havas75A,Havas75B}. These results were culminated in \cite{IarovI63}
(see also \cite{Cantri77}) were necessary and sufficient conditions for complete
separation of variables were found for more general Hamiltonians depending on
time and containing the term linear in momenta and potential, again under the
assumption that all diagonal elements of separable metrics differ from zero.

A different technique was used in \cite{KalMil80,KalMil81} (see also
\cite{Benent91}) for obtaining separable metrics including the case when the
diagonal inverse metric components include zeroes (the corresponding coordinate
lines were called there ``essential coordinates of type I'' and we call them
``coisotropic coordinates'' for brevity). However the full list of separating
action functions and conservation laws in Hamiltonian formulation was not
derived. This general problem was also attacked in \cite{Shapov80,BagObu93} and
many important examples especially with coisotropic coordinates were considered
in \cite{Obukho06,ObuOse07}. In the present paper, we solve this problem using
another technique which allows us not only to derive separable metrics but, in
addition, the full set of separating action functions and conservation laws in
the Hamiltonian formulation.

The coordinate free formulation of separability criteria in general case was
proposed in \cite{Benent97}.

It turns out that St\"ackel metrics and complete multiplicative separation of
variables in the respective Laplace--Beltrami, Helmholtz, and Schr\"odinger
equations are closely related. Namely, complete multiplicative separation of
variables in the latter equations provides sufficient conditions for complete
additive separation of variables in the Hamilton--Jacobi equation
\cite{Havas75A,Havas75B}. This observation increases the importance of separable
metrics.

Separable metrics of Lorentzian signature are of great importance in gravity
models. Usually, one assumes existence of a large symmetry group for metric to
reduce the number of independent components which allows to obtain exact
solutions of Einstein's equations in many cases. As a rule, such solutions admit
complete separation of variables in the geodesic Hamilton--Jacobi equation.
This approach can be inverted. One may assume complete separability of metric
which also reduces significantly the number of independent metric components
yielding a hope of obtaining exact solutions of gravity equations. Such
solutions are very attractive because the respective geodesic equations are
Liouville integrable, and this helps us to understand global structure of
respective space-times. This idea was successfully implemented in
\cite{Carter68} for separable metrics admitting two commutative Killing vector
fields and two indecomposable quadratic conservation laws. A large class of
exact solutions was described in this way including Schwarzschild,
Reissner--Nordstr\"om, Kerr, and many others.

It turns out that all nonzero components of separable metrics are given by
fixed functions of the set of arbitrary functions of single coordinates. This
means that vacuum Einstein's equations reduce to a system of nonlinear
{\em ordinary} differential equations. This feature enhances the hope of
obtaining exact solutions.

Complete separation of variables for the geodesic Hamilton--Jacobi equation
occurs in the presence of Killing vectors and Killing second rank tensors.
The importance of second rank Killing tensors for linear Hamiltonian systems is
discussed in \cite{Kozlov20}. Some topological properties of complete separation
of variables for Killing tensors are discussed in \cite{Kozlov95}.

In the present paper, we propose complete solution of the St\"ackel problem
for metrics of any signature in any dimensions including cases of zero diagonal
inverse metric components. Our proofs in many cases differ from others. At
the end of the paper, we give complete lists of canonical separable metrics in
two, three, and four dimensions.
\section{Separation of variables}
We consider $n$-dimensional topologically trivial manifold $\MM\approx\MR^n$
covered by a global coordinate system $x^\al$, $\al=1,\dotsc,n$. Let there be a
geodesic Hamiltonian (function on the phase space $(x,p)\in\MT^*(\MM)$)
\begin{equation}                                                  \label{anvbjd}
  H_0(x,p):=\frac12g^{\al\bt}(x)p_\al p_\bt,
\end{equation}
where $g^{\al\bt}(x)$ is the inverse metric on $\MM$ and $p_\al$ are momenta. It
is well known that it yields Hamiltonian equations for geodesics on a
(pseudo)Riemannian manifold $(\MM,g)$. All functions are assumed to be
sufficiently smooth, and we shall not mention this in what follows. Moreover we
often say simply ``metric'' instead of ``inverse metric''.

If metric is positive definite then Hamiltonian (\ref{anvbjd}) describes motion
of a point particle on Riemannian manifold $(\MM,g)$. For Lorentzian signature
metric, Hamiltonian (\ref{anvbjd}) descri\-bes worldlines of point particles on
space-time $(\MM,g)$. Both cases are of considerable interest from mathematical
and physical points of view, and many properties of such mechanical systems do
not depend on the signature of the metric. We shall consider metrics of
arbitrary signature, transition to positive definite metric being usually
trivial.

The Hamilton--Jacobi equation for the truncated action function (characteristic
Hamil\-ton function) $W(x)$ is
\begin{equation}                                                  \label{ehdggt}
  H_0\left(x,\frac{\pl W}{\pl x^\al}\right)
  =\frac12g^{\al\bt}\pl_\al W\pl_\bt W=E,\qquad E=\const.
\end{equation}

\begin{defn}
A solution of Hamilton--Jacobi equation (\ref{ehdggt}) $W(x,c)$ depending on $n$
inde\-pendent parameters (integration constants) $(c_a)\in\MV\subset\MR^n$,
$a=1,\dotsc,n$, such that
\begin{equation}                                                  \label{indbfd}
  \det\frac{\pl^2 W}{\pl x^\al\pl c_a}\ne0,
\end{equation}
is called {\em complete integral}.
\qed\end{defn}
It is not a general solution of Eq.(\ref{ehdggt}) which has functional
arbitrariness. However any solution of the Hamilton--Jacobi equation can be
obtained from a complete integral by variation of parameters (see, e.g.,
\cite{Stepan50}, Ch.~IX, \S 3). Therefore, if complete integral is known, then
the problem may be considered as solved.

Note that there are infinitely many complete integrals of the Hamilton--Jacobi
equation, if they exist. Therefore our aim is not to find all complete integrals
but only one. The domain $c\in\MV$ depends on the metric and therefore is not
specified. It should be found in every particular case, and requires some
investigation. In general, constant $E(c)$ (energy) in the Hamilton--Jacobi
equation depends on parameters. In particular, it can be considered as one of
them.

It is clear that any solution of the Hamilton--Jacobi equation is defined up to
addition of arbitrary constant. This constant cannot be chosen as one of the
parameters $c$ because condition (\ref{indbfd}) is violated. Therefore all
additive integration constants of $W$ will be dropped as inessential.

We use Latin indices $a,b,\dotsc$ for enumeration of independent parameters
$c$ though they run the same values as the Greek ones. This is done to stress
important difference: tensor components with Greek indices are transformed under
coordinate transformation whereas with Latin indices are not. For
example, functions $\pl^a W:=\pl W/\pl c_a$ are scalars while partial
derivatives $\pl_\al W:=\pl W/\pl x^\al$ are components of covector.

\begin{defn}
Coordinates $x^\al$, if they exist, are called {\em separable}, if
Hamilton--Jacobi equation (\ref{ehdggt}) admits additive separation of
variables in this coordinate system, i.e.\ the action function is given by the
sum
\begin{equation}                                                  \label{abcvde}
  W=\sum_{\al=1}^n W_\al(x^\al,c)
\end{equation}
where every summand $W_\al$ is a function of only one coordinate $x^\al$ and
parameters $c_a$, taking values in some domain $\MV\subset\MR^n$. We require
\begin{equation}                                                  \label{anvbkt}
  \det\frac{\pl^2 W}{\pl x^\al\pl c_a}
  =\det\frac{\pl^2 W_\al}{\pl x^\al\pl c_a}\ne0
\end{equation}
for all $x$ and $c$. Metric in separable coordinate system is called
{\em separable}. Functions $W_\al$ in the sum (\ref{abcvde}) are called
{\em separating}.
\qed\end{defn}

Individual summand in the action function (\ref{abcvde}) may depend only on some
part of parameters, but the whole action function depends on all $n$ parameters.

For brevity, we use notation
\begin{equation*}
  \pl^a W:=\frac{\pl W}{\pl c_a},\qquad W'_\al(x^\al,c):=\pl_\al W_\al(x^\al,c).
\end{equation*}
Then requirement (\ref{anvbkt}) is written as
\begin{equation}                                                  \label{ijdhgf}
  \det(\pl^aW'_\al)\ne0.
\end{equation}

The problem of complete separation of variables in Hamilton--Jacobi equation
(\ref{ehdggt}) is as follows. We have to describe all metrics $g^{\al\bt}(x)$,
depending only on coordinates $x$, for which there is a constant
$E(c)\ne0$ and functions $W'_\al(x^\al,c)$, depending on only one coordinate
$x^\al$ and parameters $c$ such that $\det \pl^a W'_\al\ne0$
and equation
\begin{equation}                                                  \label{abdnfh}
  g^{\al\bt}W'_\al W'_\bt=2E
\end{equation}
holds. During solution of this problem we find admissible form of separating
functions $W'_\al$ which, as we shall see, define independent envolutive
conservation laws in the corresponding Hamiltonian systems. Thus, we have to
solve functional (not differential) equation (\ref{abdnfh}) with respect to
$g^{\al\bt}(x)$ and $W'_\al(x^\al,c)$, which are supposed sufficiently smooth
both on $x$ and $c$.

All functions $W_\al$ is the sum (\ref{abcvde}) are scalars with respect to
coordinate transformations on $\MM$ though they have the coordinate index $\al$.
Covector is defined by partial derivatives $W'_\al:=\pl_\al W_\al$.
If the Hamilton--Jacobi equation admits complete separation of variables, then
the corresponding Hamiltonian equations are Liouville integrable, the solution
being given in quadratures. Requirement (\ref{ijdhgf}) means that $n$ functions
$W_\al$ are functionally inde\-pendent and can be chosen as new coordinates.
These are the action-angle coordinates in the corresponding Hamiltonian
formulation.

Consider an important example of mechanical system in which metric has no
Killing vectors but admit complete separation of variables in the
Hamilton--Jacobi equation.
\begin{exa}[\bf The Liouville system] \rm                         \label{ejsdfw}
Consider a conformally flat metric
\begin{equation}                                                  \label{ewtref}
  g_{\al\bt}=\Phi^2\eta_{\al\bt}\quad\Leftrightarrow\quad
  g^{\al\bt}:=\frac1{\Phi^2}\eta^{\al\bt},\qquad \Phi^2:=\phi_1+\dotsc+\phi_n>0,
\end{equation}
where every function $\phi_\al(x^\al)$ (no summation) depends on single
coordinate, and $\eta_{\al\bt}$ denotes (pseudo)Euclidean metric of arbitrary
signature. This metric defines the Hamilto\-nian
\begin{equation}                                                  \label{abcndp}
  H:=\frac{\eta^{\al\bt}p_\al p_\bt}{2(\phi_1+\dotsc+\phi_n)}+U(x)
  =\frac1{\Phi^2}\left(\frac12\eta^{\al\bt} p_\al p_\bt+\Theta^2\right).
\end{equation}
Here we added potential energy $U$ of special type for generality
\begin{equation*}
  U:=\frac{\Theta^2}{\Phi^2}
  =\frac{\theta_1(x^1)+\dotsc+\theta_n(x^n)}{\phi_1(x^1)+\dotsc+\phi_n(x^n)},
\end{equation*}
every function $\theta_\al(x^\al)$ being depended on single coordinate.
Hamiltonian equations of motion are
\begin{equation*}
\begin{split}
  \dot x^\al=&\frac{\eta^{\al\bt}p_\bt}{\Phi^2},
\\
  \hphantom{\qquad\ns(\al)}
  \dot p_\al=&\frac{\eta^{\bt\g}p_\bt p_\g}{2\Phi^4}\pl_\al\phi_\al
  +\frac{\Theta^2\pl_\al\phi_\al-\Phi^2\pl_\al\theta_\al}{\Phi^4}
  =\frac{H}{\Phi^2}\pl_\al\phi_\al-\frac{\pl_\al\theta_\al}{\Phi^2},
  \qquad\ns(\al),
\end{split}
\end{equation*}
were the dot denotes differentiation with respect to the evolution parameter
$\tau\in\MR$ (time), and summation on lower indices $\al$ is absent. Here and in
what follows, we denote this circumstance by special symbol $\ns(\al)$ on the
right of equations for brevity. Summation over other repeated indices is
performed as usual. The corresponding Lagrangian is
\begin{equation*}
  L:=p_\al\dot x^\al-H=\frac{\phi_1+\dotsc+\phi_n}2\eta_{\al\bt}
  \dot x^\al\dot x^\bt-\frac{\theta_1+\dotsc+\theta_n}{\phi_1+\dotsc+\phi_n}.
\end{equation*}
Sure, energy $E$ for every trajectory is conserved,
\begin{equation}                                                  \label{ewkguy}
  E:=H\big|_{\text{trajectory}}=\const.
\end{equation}

Multiply Hamilton--Jacobi equation (\ref{abdnfh}) by the conformal factor:
\begin{equation*}
  \eta^{\al\bt}W'_\al W'_\bt+2\Theta^2=2E\Phi^2.
\end{equation*}
Separation of variables for the Liouville system happens in the way
\begin{equation}                                                  \label{eliojk}
  \hphantom{\qquad\qquad\ns(\al)}
  \eta^{\al\al}W^{\prime2}_\al+2\theta_\al-2E\phi_\al=c_\al,\qquad
  c_\al\in\MR,\quad\forall \al,
  \qquad\qquad\ns(\al).
\end{equation}

The action function has the form
\begin{equation*}
  W=\sum_{\al=1}^n W_\al,\qquad W_\al:=\int\!\!dx^\al
  \sqrt{\eta^{\al\al}\big(c_\al+2E\phi_\al-2\theta_\al\big)}.
\end{equation*}
Certainly, the radicand is assumed to be positive. This complete integral
satisfies the Hamilton--Jacobi equation multiplied by the conformal factor:
\begin{equation}                                                  \label{abwswd}
  \dl^{\al\bt}\pl_\al W\pl_\bt W=2E\Phi^2,
\end{equation}
where multipliers $\eta^{\al\al}$ are included in the action function $W$.

In the Hamiltonian formalism, we have $n$ quadratic conservation laws
\begin{equation}                                                  \label{eddgtr}
  \eta^{\al\al}p_\al^2+2\theta_\al-2E\phi_\al=c_\al.
\end{equation}
This can be easily checked by straightforward computation:
\begin{equation*}
  \dot c_\al=2\eta^{\al\al}p_\al\dot p_\al+2\dot\theta_\al-2E\dot \phi_\al
  =\eta^{\al\al}p_\al\frac{2H\pl_\al \phi_\al}{\Phi^2}
  -2E\dot x^\al\pl_\al \phi_\al
  =\eta^{\al\al}p_\al\frac{2(H-E)}{\Phi^2}\pl_\al \phi_\al=0,
\end{equation*}
because $H=E$ along every trajectory. For nontrivial functions $\theta$ and
$\phi$ these conservation laws are indecomposable, parameters satisfying the
condition $c_1+\dotsc+c_n=0$. Indeed, sum all conservation laws (\ref{eddgtr})
taking into account the definition of Hamiltonian (\ref{abcndp}):
\begin{equation}                                                  \label{eefdrr}
  c_1+\dotsc+c_n=(2H-2E)(\phi_1+\dotsc+\phi_n)=0.
\end{equation}
This is the sole relation between parameters. It shows that only $n-1$
conservation laws among (\ref{eliojk}) are independent. The last independent
conservation law is also quadratic
\begin{equation}                                                  \label{essghk}
  \frac1{\Phi^2}\big(\eta^{\al\al}p_\al^2+2\Theta\big)=2E.
\end{equation}
It depends on all momenta and coordinates. The complete set of independent
parameters is, for example, $c_1,\dotsc,c_{n-1},2E$.

Thus, all canonical pairs of variables are separated. The particular feature of
this separation is that conformally flat metric (\ref{ewtref}) has no Killing
vector for nonconstant functions $\phi_\al(x^\al)$. Besides, variables were
separated only after multiplication of the Hamil\-ton--Jacobi equation by the
conformal factor. Separation of variables for the Liouville system does not
depend on the metric signature.
\qed\end{exa}
\section{Separation of variables and Hamiltonian formalism}
In this section, we show what happens in Hamiltonian formalism for complete
separation of variables in the Hamilton--Jacobi equation. Remind that we are
considering Hamiltonian (\ref{anvbjd}) yielding equations of motion
\begin{equation}                                                  \label{abcvdr}
\begin{split}
  \dot x^\al=&[x^\al,H_0]=g^{\al\bt}p_\bt,
\\
  \dot p_\al=&[p_\al,H_0]=\frac12\pl_\al g_{\bt\g}p^\bt p^\g.
\end{split}
\end{equation}
Suppose that relations
\begin{equation}                                                  \label{aocvfd}
  \hphantom{\qquad\qquad\ns(\al)}
  G_\al(x^\al,p_\al,c):=p_\al-W'_\al(x^\al,c)=0,\qquad\forall\al,
  \qquad\qquad\ns(\al),
\end{equation}
hold, where $(W'_\al)$ is an arbitrary solution of the Hamilton--Jacobi equation
(\ref{abdnfh}) depending on $n$ independent parameters. Due to inequality
(\ref{ijdhgf}), these equations can be locally solved with respect to $c$:
\begin{equation}                                                  \label{anbfgt}
  c_a=F_a(x,p),
\end{equation}
where $F_a$ are some functions on canonical variables $x,p$ only.
\begin{theorem}                                                   \label{thdkwe}
Relations (\ref{anbfgt}) hold along every trajectory of the Hamiltonian system.
These first integrals of Hamiltonian equations are in involution.
\end{theorem}
\begin{proof}
Compute
\begin{equation}                                                  \label{anbftg}
  \dot c_a=\frac{\pl F_a}{\pl x^\al}\dot x^\al
  +\frac{\pl F_a}{\pl p_\al}\dot p_\al.
\end{equation}
To find partial derivatives of functions $F_a$, we do the following. Take the
differential of Eq.(\ref{aocvfd}):
\begin{equation*}
  \hphantom{\qquad\qquad\ns(\al)}
  dp_\al-\pl_\al W'_\al dx^\al-\pl^a W'_\al dc_a=0,
  \qquad\qquad\ns(\al).
\end{equation*}
These equations can be solved with respect to $dc_a$:
\begin{equation}                                                  \label{abbdhf}
  dc_a=\sum_{\al=1}^n\left(\pl^a W'_\al\right)^{-1}
  \big(dp_\al-\pl_\al W'_\al dx^\al\big),
\end{equation}
because of inequality (\ref{ijdhgf}). On the other hand, the differential of
Eq.(\ref{anbfgt}) yields
\begin{equation*}
  dc_a=\frac{\pl F_a}{\pl x^\al}dx^\al+\frac{\pl F_a}{\pl p_\al}dp_\al.
\end{equation*}
Comparing Eqs.(\ref{anbftg}) and (\ref{abbdhf}), we conclude that
\begin{equation*}
  \hphantom{\qquad\qquad\ns(\al)}
  \frac{\pl F_a}{\pl x^\al}=-(\pl^a W'_\al)^{-1}\pl_\al W'_\al,\qquad
  \frac{\pl F_a}{\pl p_\al}=(\pl^a W'_\al)^{-1},
  \qquad\qquad\ns(\al).
\end{equation*}
Now Eq.(\ref{anbftg}) yields:
\begin{equation*}
  \dot c_a=\sum_{\al=1}^n(\pl^a W'_\al)^{-1}\left(-\pl_\al W'_\al p_\al
  +\frac12\pl_\al g_{\bt\g}p^\bt p^\g\right),
\end{equation*}
where Hamiltonian equations (\ref{abcvdr}) are used for exclusion of derivatives
with respect to evolution parameter $\tau$.

On the other hand, differentiate Hamilton--Jacobi equation (\ref{abdnfh}) with
respect to $x^\g$:
\begin{multline*}
  \pl_\g g^{\al\bt} W'_\al W'_\bt+2g^{\al\bt}W'_\al\pl_\g W'_\bt=
  -g^{\al\dl}g^{\bt\e}\pl_\g g_{\dl\e}W'_\al W'_\bt+2W^{\prime\al}\pl_\g W'_\al=
\\
  =-\pl_\g g_{\dl\e}W^{\prime\dl}W^{\prime\e}+2W^{\prime\al}\pl_\g W'_\al=
  -\pl_\g g_{\dl\e}p^\dl p^\e+2p^\g\pl_\g W'_\g=0,
\end{multline*}
where Eq.(\ref{aocvfd}) is used. Therefore $\dot c_a=0$, and the first statement
of the theorem is proved.

Now compute the Poisson bracket of two conservation laws with indices $a\ne b$
\begin{multline*}
  [F_a,F_b]=\frac{\pl F_a}{\pl x^\al}\frac{\pl F_b}{\pl p_\al}
  -\frac{\pl F_a}{\pl p_\al}\frac{\pl F_b}{\pl x^\al}=
\\
  =-\sum_{\al=1}^n\Big((\pl^a W'_\al)^{-1}\pl_\al W'_\al(\pl^b W'_\al)^{-1}-
  (\pl^a W'_\al)^{-1}(\pl^b W'_\al)^{-1}\pl_\al W'_\al\Big)=0.
\end{multline*}
This proves that conservation laws are in involution.
\end{proof}

We see that, if variables are completely separated in the Hamilton--Jacobi
equation, then the corresponding Hamiltonian system admits $n$ independent
conservation laws (\ref{anbfgt}) which are in involution, and, consequently, it
is Liouville integrable. This statement is proved for the geodesic Hamiltonian
$H_0$. Moreover, it is valid in a general case.

If summands $W'_\al(x^\al,c)$ are found, then conservation laws (\ref{anbfgt})
are obtained by solution of Eqs.~(\ref{aocvfd}) with respect to parameters in
explicit form. Luckily, this can be always done for geodesic Hamiltonian.

If a complete integral of the Hamilton--Jacobi equation is known, then one can
always go to the action-angle coordinates at least implicitly.
\begin{theorem}
Let us perform the canonical transformation $(x,p)\mapsto(X,P)$ with generating
function $S_2(x,P):=W(x,P)$, where momenta are substituted instead of
parameters: $(c_a)\mapsto(P_a)$. Then Hamiltonian equations for new variables
take the form
\begin{equation}                                                  \label{ahsgtr}
\begin{split}
  \dot X^a=&\frac{\pl H}{\pl P_a}=\frac{\pl H}{\pl p_\al}\pl^a\pl_\al W,
\\
  \dot P_a=&0,\qquad \Rightarrow\qquad P_a=c_a=\const,
\end{split}
\end{equation}
where the substitution $x=x(X,P)$ and $p=p(X,P)$ is performed on the right hand
side.
\end{theorem}
\begin{proof}
If the generating function of the canonical transformation depends on old
coordinates and new momenta, then old momenta and new coordinates are given by
formulae (see, e.g., \cite{Goldst50})
\begin{equation}                                                  \label{aggwer}
  p_\al=\pl_\al W,\qquad X^a=\pl^a W.
\end{equation}
The last equality can be solved with respect to old coordinates at least locally
due to inequality (\ref{indbfd}), and we find $x=x(X,P)$. Substitution of this
solution in the first equality defines $p=p(X,P)$.

Take the differential of Eqs.~(\ref{aggwer}):
\begin{equation}                                                  \label{abbvgd}
\begin{split}
  dp_\al=&\pl^2_{\al\bt}Wdx^\al+\pl^a\pl_\al W dP_a,
\\
  dX^a=&\pl^a\pl_\al W dx^\al+\pl^a\pl^bWdP_b.
\end{split}
\end{equation}
This defines differentials
\begin{equation}                                                  \label{anbbdg}
\begin{split}
  dx^\al=&(\pl^a\pl_\al W)^{-1}(dX^a-\pl^a\pl^b WdP_b),
\\
  dp_\al=&\pl^2_{\al\bt} W(\pl^a\pl_\al W)^{-1}dX^a
  +\big[\pl^a\pl_\al W-\pl^2_{\al\bt}W(\pl^b\pl_\bt W)^{-1}\pl^a\pl^bW\big]dP_a.
\end{split}
\end{equation}
Consequently, partial derivatives are
\begin{equation}                                                  \label{abxvsf}
\begin{split}
  \frac{\pl x^\al}{\pl X^a}=&(\pl^a\pl_\al W)^{-1},
\\
  \frac{\pl x^\al}{\pl P_a}=&-(\pl^b\pl_\al W)^{-1}\pl^a\pl^b W,
\\
  \frac{\pl p_\al}{\pl X^a}=&\pl^2_{\al\bt}W(\pl^a\pl_\bt W)^{-1},
\\
  \frac{\pl p_\al}{\pl P_a}=&\pl^a\pl_\al W
  -\pl^2_{\al\bt}W(\pl^b\pl_\bt W)^{-1}\pl^a\pl^b W.
\end{split}
\end{equation}

On the other hand the Hamilton--Jacobi equation implies
\begin{equation*}
  H\left(x,\frac{\pl W(x,P)}{\pl x}\right)=2E(P).
\end{equation*}
Differentiate this equality by $x^\al$:
\begin{equation}                                                  \label{akkfyf}
  \frac{\pl H}{\pl x^\al}+\frac{\pl H}{\pl p_\bt}\frac{\pl p_\bt}{\pl x^\al}=
  \frac{\pl H}{\pl x^\al}+\frac{\pl H}{\pl p_\bt}\pl_{\al\bt}W=0.
\end{equation}
Now we see that equalities
\begin{equation*}
  \frac{\pl H}{\pl X^a}\equiv0,\qquad
  \frac{\pl H}{\pl P_a}=\frac{\pl H}{\pl p_\al}\pl^a\pl_\al W,
\end{equation*}
hold where $H=H\big(x(X,P),p(X,P)\big)$.
\end{proof}

Separable coordinates, when they exist, are not uniquely defined, and there is
a functional arbitrariness related to coordinate transformations. To isolate it,
we give
\begin{defn}
Two separable coordinate systems $x$ and $X$ on $\MM$ are {\em equivalent}, if
there is a canonical transformation $(x,p)\mapsto(X,P)$ of respective
Hamiltonian systems, such that new coordinates $X(x)$ depend only on old ones
but not on momenta $p$. Moreover, separable coordinate systems are equivalent
in case when parameters are related by nondegenerate transformation
$c\mapsto\tilde c(c)$ which do not involve coordinates. Manifold $\MM$ with
metric $g$ which admits separable coordinates for Eq.~(\ref{ehdggt}) in some
neighbourhood of each point related by this equivalence relation in overlapping
domains is called {\em  St\"ackel}.
\qed\end{defn}

Right hand sides of conservation laws (\ref{anbfgt}) are scalars under canonical
transformations:
\begin{equation*}
  F(x,p)\mapsto \tilde F_a(X,P)=F_a\big(x(X),p(X,P)\big).
\end{equation*}

Conservation laws are transformed into conservation laws and their involution
survives under arbitrary canonical transformation because the latter preserves
the Poisson bracket. We shall use canonical transformations in what follows to
bring separable metrics to the most simple (canonical) forms and eliminate
inessential functions.

There may occur inequivalent separable coordinate systems for one and the same
metric. In next sections, we obtain necessary and sufficient conditions for
existence of separable coordinate systems and present explicit form of
canonical separable metrics in each class of equivalent separable metrics.
Before solution of the St\"ackel problem, we consider the simple instructive
example.
\begin{exa}\rm                                                    \label{esgwfr}
Consider Euclidean space $\MR^n$ with metric $\dl^{ab}$,
$a,b=1,\dotsc,n$ in Cartesian coordinates $y^a$. The Hamilton--Jacobi equation
\begin{equation*}
  \dl^{ab}W'_aW'_b=2E
\end{equation*}
admits complete separation of variables
\begin{equation*}
  W'_a=c_a=\const,
\end{equation*}
the energy $2E=c^ac_a$ being the dependent parameter. The action function is
\begin{equation}                                                  \label{ajcvdf}
  W=y^ac_a,
\end{equation}
where we dropped inessential additive integration constant. It depends on
$n$ independent parameters $c_1,\dotsc,c_n$. Consequently, action function
(\ref{ajcvdf}) is the complete integral of the Hamilton--Jacobi equation, and
the problem is solved. Complete separation of variables (\ref{ajcvdf}) implies
$n$ conservation laws in the Hamiltonian formulation:
\begin{equation*}
  p_a=c_a,
\end{equation*}
which means conservation of all momenta. We see that Cartesian coordinate
system in Euclidean space is separable. Vector fields $\pl_a$ are Killing
vector fields. At the same time most of curvilinear coordinates in $\MR^n$ are
not separable.

Indeed, let us go to curvilinear coordinates $y^a\mapsto x^\al(y)$,
$\al=1,\dotsc,n$. Then inverse functions $y^a(x)$ define vielbein (Jacobi
matrices) and functions $W'_\al(x,c)$ are
\begin{equation*}
  W'_\al=e_\al{}^a c_a,\qquad e_\al{}^a:=\pl_\al y^a,
\end{equation*}
because they are covector components. It is necessary and sufficient for new
coordinates to be separable that vielbein components $e_\al{}^a(x^\al)$ depend
only on one coordinate $x^\al$. It means that transition functions must be
specific
\begin{equation*}
  y^a=\sum_{\al=1}^n k_\al{}^a(x^\al),
\end{equation*}
where all elements of each row of the matrix $k_\al{}^a$ depend only on one
coordinate $x^\al$. Then the vielbein is
\begin{equation*}
  e_\al{}^a(x^\al)=\pl_\al k_{\al}{}^a=:k'_\al{}^a
\end{equation*}
with the sole condition $\det k'_\al{}^a\ne0$. Now we compute metric
$g^{\al\bt}$, and variables in the Hamilton--Jacobi equation are separated:
\begin{equation*}
  \hphantom{\qquad\qquad\ns(\al)}
  W=\sum_{\al=1}^n W_\al,\qquad W_\al:=\int\!\!dx^\al e_\al{}^a(x^\al)c_a,
  \qquad\qquad\ns(\al).
\end{equation*}
In the Hamiltonian formulation, we have $n$ independent conservation laws:
\begin{equation*}
  p_\al e^\al{}_a(x)=c_a.
\end{equation*}
Note that the inverse vielbein components $e^\al{}_a(x)$, in general, depend on
all coordinates.

After the inverse transformation $x^\al\mapsto y^a$, the conservation laws take
simple form $p_a=c_a$ which correspond to $n$ commuting Killing vectors $\pl_a$
(translations along Cartesian coordinates).

All coordinate systems obtained in this way and only they (this will be proved
in section \ref{sbskdy}) are equivalent. To see this, we make canonical
transformation $(x^\al,p_\al)\mapsto(y^a,p_a)$ with generating function
depending on old coordinates and new momenta
\begin{equation*}
  S_2(x^\al,p_a):=\sum_{\al=1}^n k_\al{}^a(x^\al)p_a.
\end{equation*}
Then
\begin{equation}                                                  \label{andfmj}
  p_\al=\frac{\pl S_2}{\pl x^\al}=k'_\al{}^ap_a=e_\al{}^ap_a,\qquad
  y^a=\frac{\pl S_2}{\pl p_a}=\sum_{\al=1}^n k_\al{}^a,
\end{equation}
and we are led to previous formulae. Thus, if we solve functional equation
(\ref{abdnfh}), then its general solution contains many arbitrary functions.
In our case, there are $n^2$ inessential functions $k_\al{}^a(x^\al)$
describing transformations between equivalent separable coordinate systems. We
see that the problem can be essentially simplified by making canonical
transformations. Therefore it is sufficient to choose the most simple separable
metric in each class of equivalent metrics which is called canonical. In the
present case, it is the Euclidean metric, all other equivalent separable metrics
are related to it by suitable canonical transformation. We note that
transformation of momenta (\ref{andfmj}) is linear. Therefore linear and
quadratic conservation laws are transformed into linear and quadratic ones,
respectively.

At the same time, separation of variables in Euclidean space may happen in a
different way. Consider Euclidean plane in polar coordinates
$(y^a)=(r,\vf)\in\MR^2$. In Cartesian coordinates, we have two linear
conservation laws
\begin{equation*}
  p_1=c_1,\qquad p_2=c_2,
\end{equation*}
corresponding to two commuting Killing vectors $\pl_1$ and $\pl_2$. In polar
coordinates,
\begin{equation*}
  x:=y^1=r\cos\vf,\qquad y:=y^2=r\sin\vf.
\end{equation*}
metric and its inverse are
\begin{equation*}
  g_{\al\bt}=\begin{pmatrix} 1 & 0 \\ 0 & r^2 \end{pmatrix},\qquad
  g^{\al\bt}=\begin{pmatrix} 1 & 0 \\ 0 & \frac1{r^2} \end{pmatrix}.
\end{equation*}
The Hamilton--Jacobi equation is
\begin{equation*}
  W^{\prime2}_r+\frac1{r^2}W^{\prime2}_\vf=\dl^{ab}c_ac_b=2E.
\end{equation*}
Variables in this equation are separated as follows
\begin{equation}                                                  \label{anbvfs}
  W'_\vf=c_\vf,\qquad W^{\prime2}_r+\frac{c_\vf^2}{r^2}=2E.
\end{equation}
The first equality corresponds to rotational invariance (Killing vector
$\pl_\vf$), and the second is the conservation of energy. The action function is
now
\begin{equation*}
  W=\int\!\!dr\sqrt{2E-\frac{c_\vf^2}{r^2}}+\vf c_\vf.
\end{equation*}
It is the complete integral of the Hamilton--Jacobi equation. If $E>0$, then
this expression makes sense only for $r^2>c_\vf^2/2E$.

In the Hamiltonian formulation, we have two conservation laws
\begin{equation*}
  p_\vf=c_\vf,\qquad p_r^2+\frac{c_\vf^2}{r^2}=2E,
\end{equation*}
These are conservation of angular momentum and energy.
Thus variables are also separated in polar coordinates, but now we have one
linear and one indecomposable quadratic conservation laws.

Though variables are separated both in Cartesian and polar coordinates, these
coordinate systems are not equivalent in the sense of our definition. Indeed,
the transformation of coordinate does exist but it is not canonical because
conservation laws do not transform into conservation laws.

In polar and Cartesian coordinates, parameters $(E,c_\vf)$ and $(c_1,c_2)$
constitute full sets of parameters, respectively, for the complete action
functions. We have $2E=c_1^2+c_2^2$, but parameter $c_\vf$ cannot be expressed
entirely through $c_1$ and $c_2$. Really, momenta in polar and Cartesian
coordinates are related by transformation
\begin{equation*}
  p_r=\cos\vf\, p_1+\sin\vf\, p_2,\qquad p_\vf=-r\sin\vf\, p_1+r\cos\vf\, p_2.
\end{equation*}
It implies expression for the angular parameter
\begin{equation*}
  c_\vf=-r\sin\vf\, c_1+r\cos\vf\, c_2=-yc_1+xc_2.
\end{equation*}
That is transformation of parameters $(c_1,c_2)\mapsto(E,c_\vf,x,y)$ necessarily
include coordinates. We see that transformation between full sets of independent
parameters may include coordinates, and separation of variables does depend on
our choice.
\qed\end{exa}

Considered examples teach us the following. First, a general solution of the
St\"ackel problem contains many functions parameterizing transformations between
equivalent separable coordinate systems. In configuration space we have
ordinary coordinate transformations which do not alter geometric invariants,
e.g.\ scalar curvature or squares of curvature tensor. These functions are
inessential, and can be eliminated by suitable canonical transformation. Note
that these transformations do not exhaust the whole group of diffeomorphisms.
Second, transformation of one full set of parameters to the other
$(c_a)\mapsto(\tilde c_a,x)$ may include coordinates. In this case, separation
of variables holds differently, and corresponding separable coordinates are not
equivalent.
\section{Linear conservation laws}                                \label{sbskdy}
Example \ref{esgwfr} shows that separation of variables and conservation laws
depend on the choice of complete set of parameters. The right hand side of the
Hamilton--Jacobi equation (\ref{abdnfh}) depends on one constant $E$ which
therefore is distinguished. In this section, we choose ``symmetric'' set of
independent parameters which does not include $E$. This is done as follows.

Fix a point $\Sp\in\MM$ and introduce parameters
\begin{equation*}
  c_a:=W'_{\al=a}\big|_\Sp.
\end{equation*}
Here we change indices into Latin ones, $\al\mapsto a$ to stress that the set of
parameters $(c_a)$ after coordinate transformations is multiplied by the inverse
Jacobi matrix at a fixed point $\Sp$. The Hamilton--Jacobi equation at this
point is
\begin{equation*}
  g^{ab}c_a c_b=2E,\qquad g^{ab}:=g^{\al\bt}\big|_{\Sp,\al=a,\bt=b},
\end{equation*}
where subscript $\Sp$ means restriction of metric to this point. That is $E(c)$
becomes the dependent parameter. Now we choose the collection $(c_a)$,
$a=1,\dotsc,n$, as the complete set of independent parameters. Moreover metric
$g^{ab}$ at point $\Sp$ can be transformed into diagonal form with $\pm1$ on the
diagonal depending on the signature of the metric by suitable linear coordinate
transformation. Denote this matrix by
\begin{equation}                                                  \label{abdfre}
  \eta^{ab}:=\diag(\underbrace{1,\dotsc,1,}_r\underbrace{-1,\dotsc,-1}_s),
  \qquad r+s=n,
\end{equation}
where the pair $(r,s)$ designates the signature of the metric. Then the
Hamilton--Jacobi equation takes the ``symmetrical'' form
\begin{equation}                                                  \label{anvbfd}
  g^{\al\bt}W'_\al W'_\bt=\eta^{ab}c_a c_b,
\end{equation}
where each function $W'_\al(x^\al,c)$ depends on single coordinate $x^\al$ and,
probably, on the full set of parameters $c$. The left hand side of this
equation is a geometric invariant. Therefore we can regard metric components
$g^{\al\bt}$ and functions $W'_\al$ as transforming by usual tensor
transformation rules under coordinate changes while all quantities with Latin
indices remain unchanged.

The geometric meaning of this set of parameters is as follows. One and only one
geodesic goes through each point $\Sp$ in every direction. Therefore the set
of parameters $(x_\Sp,c)$ yields the Cauchy data and uniquely defines geodesic
line in some neighborhood of point $\Sp$. In the Hamiltonian formulation, the
set $(c_a)$ represent momenta at a fixed point $\Sp$.

Now we derive necessary and sufficient conditions for complete separation of
variables for ``symmetric'' choice of parameters in Hamilton--Jacobi equation
(\ref{anvbfd}). There is a global vielbein $e^\al{}_a(x)$ on every
topologically trivial manifold $\MM\approx\MR^n$:
\begin{equation*}
  g^{\al\bt}=e^\al{}_ae^\bt{}_b\eta^{ab},
\end{equation*}
which is defined up to local $\MO(r,s)$ (pseudo)rotations which include
coordinate reflections. The inverse vielbein is denoted by $e_\al{}^a$
($e_\al{}^be^\bt{}_b=\dl_\al^\bt$, $e_\al{}^ae^\al{}_b=\dl^a_b$).
Then Eq.~(\ref{anvbfd}) is rewritten as
\begin{equation*}
  \eta^{ab}W'_aW'_b=\eta^{ab}c_ac_b,\qquad W'_a:=e^\al{}_aW'_\al.
\end{equation*}
It implies that covectors $W'_a$ and $c_a$ are necessarily related by some
(pseudo)rotational matrix $S\in\MO(r,s)$:
\begin{equation*}
  W'_a(x)=S_a{}^b(x)c_b.
\end{equation*}
multiply this expression by inverse vielbein $e_\al{}^a$ and obtain
\begin{equation*}
  W'_\al=\tilde e_\al{}^ac_a,\qquad\text{where}\qquad
  \tilde e_\al{}^a:=e_\al{}^bS_b{}^a.
\end{equation*}
For complete separation of variables in the Hamilton--Jacobi equation, it is
necessary and sufficient that the right hand side of this relation depends
solely on $x^\al$. Drop the tilde sign and formulate the result.
\begin{theorem}                                                   \label{tbnfku}
Variables in Hamilton--Jacobi equation (\ref{anvbfd}) with ''symmetric`` set of
independent parameters $(c_a)$ are completely separated if and only if there
exists the vielbein $e_\al{}^a(x^\al)$ such that its components with fixed $\al$
depend solely on one coordinate $x^\al$ for all values of index $a$. Then every
term in action function (\ref{abcvde}) is a primitive
\begin{equation}                                                  \label{aicndg}
  \hphantom{\qquad\qquad\ns(\al)}
  W_\al(x_\al,c):=\int\!\!dx^\al e_\al{}^a(x^\al)c_a,
  \qquad\qquad\ns(\al).
\end{equation}
In the Hamiltonian formulation, there are $n$ linear in momenta conservation
laws
\begin{equation}                                                  \label{acbdgf}
  e^\al{}_a(x)p_\al=c_a
\end{equation}
for all $a$.
\end{theorem}
\begin{cor}
Covariant components of the separable metric for Hamilton--Jacobi equation
(\ref{anvbfd})
\begin{equation*}
  g_{\al\bt}(x^\al,x^\bt)=e_\al{}^a(x^\al)e_\bt{}^b(x^\bt)\eta_{ab}
\end{equation*}
depend only on two coordinates.
\qed\end{cor}
The proof is evident. Sure, components of the inverse metric $g^{\al\bt}(x)$ and
vielbein $e^\al{}_a(x)$ depend on all coordinates in general. That is
conservation laws (\ref{acbdgf}) are linear in momenta but depend on all
coordinates.

\begin{exa}\rm
Complicate the problem considered in example \ref{esgwfr}. Suppose that
rotationally symmetric metric is given in polar coordinates on a plane
$(r,\vf)\in\MR^2$
\begin{equation}                                                  \label{eiyrfd}
  g_{\al\bt}=\begin{pmatrix} f^2 & 0 \\ 0 & r^2 \end{pmatrix},\qquad
  g^{\al\bt}=\begin{pmatrix} \frac1{f^2} & 0 \\ 0 & \frac1{r^2}\end{pmatrix},
\end{equation}
where $f(r)>0$ is some function of radius only. Then the Hamilton--Jacobi
equation in polar coordinates is
\begin{equation*}
  \frac{W^{\prime2}_r}{f^2}+\frac{W^{\prime2}_\vf}{r^2}=2E,
\end{equation*}
and variables are completely separated:
\begin{equation*}
  W^{\prime}_\vf=c_\vf,\qquad\frac{W^{\prime2}_r}{f^2}+\frac{c_\vf^2}{r^2}=2E.
\end{equation*}
Parameters $c_\vf$ and $E$, as before, correspond to conservation of the angular
momentum and the energy.

Let us look what happens in Cartesian coordinates $(x,y)\in\MR^2$. In polar
coordinates vielbein can be chosen in diagonal form
\begin{equation*}
  e_\al{}^a=\begin{pmatrix} f & 0 \\ 0 & r \end{pmatrix},\qquad
  e^\al{}_a=\begin{pmatrix} \frac1f & 0 \\ 0 & \frac1r \end{pmatrix}.
\end{equation*}
In Cartesian coordinates it is written as
\begin{equation*}
  e_x{}^a=\left(\frac xr f,\,-\frac yr\right),\qquad
  e_y{}^a=\left(\frac yr f,\,\frac xr\right).
\end{equation*}
Vielbein components $e_x{}^a$ explicitly depend on $y$, and variables are not
separated. According to theorem \ref{tbnfku}, separable coordinates exist if
there is the vielbein
\begin{equation*}
  \tilde e_\al{}^a=e_\al{}^b S_b{}^a,\qquad S_b{}^a=\begin{pmatrix} \cos\om &
  -\sin\om \\ \sin\om & ~~\cos\om \end{pmatrix}\in\MS\MO(2),
\end{equation*}
where $\om(x,y)$ is some rotational angle, such that vielbein components
$\tilde e_x{}^a(x)$ and $\tilde e_y{}^a(y)$ depend only on $x$ and $y$,
respectively.
\begin{prop}
If $f^2\nequiv1$, then there is no rotational matrix $S$ such that components
of the first row of vielbein $\tilde e_x{}^a$ depend solely on $x$ for all
$a=1,2$, and components of the second row $\tilde e_y{}^a$ depend solely on $y$.
\end{prop}
\begin{proof}
Two components of vielbein in Cartesian coordinates after the rotation are
\begin{equation*}
\begin{split}
  \tilde e_x{}^1=&~~\frac xrf\cos\om-\frac yr\sin\om,
\\
  \tilde e_x{}^2=&-\frac xrf\sin\om-\frac yr\cos\om.
\end{split}
\end{equation*}
Their independence on $y$ is given by equations
\begin{align}                                                     \label{anskhg}
  \pl_y\tilde e_x{}^1=&-\frac{xy}{r^3}f\cos\om+\frac{xy}{r^2}f'\cos\om
  -\frac xrf\sin\om\pl_y\om-\frac{x^2}{r^3}\sin\om-\frac yr\cos\om\pl_y\om=0,
\\                                                                \label{anbswe}
  \pl_y\tilde e_x{}^2=&~~\frac{xy}{r^3}f\sin\om-\frac{xy}{r^2}f'\sin\om
  -\frac xrf\cos\om\pl_y\om-\frac{x^2}{r^3}\cos\om+\frac yr\sin\om\pl_y\om=0,
\end{align}
where $f':=df/dr$. Take the linear combination of theses equations
\begin{equation*}
  (\ref{anskhg})\sin\om+(\ref{anbswe})\cos\om=
  -\frac xrf\pl_y\om-\frac{x^2}{r^3}=0\qquad\Rightarrow\qquad
  f\pl_y\om=-\frac x{r^2}
\end{equation*}
and substitute in Eq.~(\ref{anskhg}). Finally, we obtain equation for $f$
\begin{equation*}
  rff'-f^2+1=0.
\end{equation*}
Its general solution is
\begin{equation*}
  f^2=1+Cr^2,\qquad C=\const.
\end{equation*}

Now we write the independence of the second row of the vielbein on $x$,
$\pl_x\tilde e_y{}^a=0$, and make the same calculations as for
Eqs.~(\ref{anskhg}) and (\ref{anbswe}). Then we obtain
\begin{equation*}
  f\pl_x\om=\frac y{r^2}
\end{equation*}
with the same function $f$. Consequently, the rotation angle is defined by the
system of equation
\begin{equation*}
  \pl_x\om=\frac y{r^2 f},\qquad \pl_y\om=-\frac x{r^2f}.
\end{equation*}
It is easily checked that integrability conditions are not fulfilled for
$f^2\ne1$, and therefore function $\om(x,y)$ does not exist. Therefore variables
for metric (\ref{eiyrfd}) are separated in Cartesian coordinates only for
$f^2\equiv1$.
\end{proof}
Thus variables in the Hamilton--Jacobi equation for metric (\ref{eiyrfd}) are
completely separated in polar coordinates but not in Cartesian. This is a
natural result because metric (\ref{eiyrfd}) is not translationally invariant.
\qed\end{exa}

The moral of this example is that if the Hamilton--Jacobi equation is not
separable for ''symmetric`` choice of parameters, then it does not mean that
separability does not happen for other choices of independent parameters.
Consequently, theorem \ref{tbnfku} does not exhaust all possibilities of
variable separations.

Conservation laws (\ref{acbdgf}) can be significantly simplified by canonical
transformation $(x^\al,p_\al)\mapsto(X^a,P_a)$ with generating function
depending on old coordinates $(x^\al)$ and new momenta $(P_a)$
\begin{equation}                                                  \label{ancbxk}
  \hphantom{\qquad\qquad\ns(\al)}
  S_2:=\sum_{\al=1}^n I_\al,\qquad I_\al:=\int\!\!dx^\al\,e_\al{}^a(x^\al)P_a,
  \qquad\qquad\ns(\al),
\end{equation}
where $I_\al$ are primitives. Then
\begin{equation*}
  p_\al=\frac{\pl S_2}{\pl x^\al}=e_\al{}^aP_a,\qquad
  X^a=\frac{\pl S_2}{\pl P_a}=\sum_{\al=1}^\Sn\int\!\!dx^\al e_\al{}^a.
\end{equation*}
This is exactly the coordinate transformation $(x^\al)\mapsto\big(X^a(x)\big)$.
Indeed, the Jacobi matrix is
\begin{equation*}
  \pl_\al X^a=e_\al{}^a,
\end{equation*}
and momenta components are transformed as covariant vector. In new coordinates,
the Hamiltonian is
\begin{equation*}
  H_0=\frac12\eta^{ab}P_aP_b.
\end{equation*}
It implies $n$ independent linear conservation laws
\begin{equation}                                                  \label{azvxcs}
  P_a=c_a
\end{equation}
(all momenta are conserved). Since metric components are constant in these
coordinates then there is the maximal number $n$ of linearly independent
commuting Killing vectors $\pl_a$ on $\MM$. In this way, we proved
\begin{theorem}                                                   \label{tdjkfh}
Variables in the Hamilton--Jacobi equation (\ref{anvbfd}) with the ``symmetric''
choice of parameters are completely separable if and only if there exist
$n$ linearly independent at each point commuting Killing vector fields on
$(\MM,g)$. Then there is a coordinate system in which metric components are
constant, and conservation laws have form (\ref{azvxcs}).
\end{theorem}

This result is not surprising. Indeed, if vielbein components $e_\al{}^a(x^\al)$
depend on single coordinate $x^\al$ for all $a$, then curvature tensor vanishes.
These calculations are most easily performed in Cartan variables. The
nonholonomicity components are identically zero
\begin{equation*}
  c_{\al\bt}{}^a:=-\pl_\al e_\bt{}^a+\pl_\bt e_\al{}^a\equiv0.
\end{equation*}
Therefore the corresponding $\MS\MO(r,s)$ connection is also zero, and,
consequently, curvature vanishes. Therefore manifold $\MM$ is locally
(pseudo)Euclidean. This is enough for complete separation of variables in
Cartesian coordinates. Thus, if complete separation of variables occurs for
``symmetrical'' choice of independent parameters, then manifold $\MM$ is
necessarily locally flat.
\section{Quadratic conservation laws and coisotropic coordinates} \label{sajwuy}
Now we consider another possibility of complete variables separation of
variables in the Hamilton--Jacobi equation. Remember that our aim is not to find
all complete integrals, which are infinitely many, but to find at least one of
them for a given type of separable metric. Therefore our strategy is the
following.
\begin{enumerate}
\item Describe all possible types of separable metrics.
\item Choose independent parameters.
\item Specify the class of separating functions in each case.
\item Solve the functional Hamilton--Jacobi equation.
\item Choose the canonical separable metric in every equivalence class,
find the complete integral and the full set of conservation laws.
\item Prove that any additional conservation law is functionally dependent
on the full set of previously obtained conservation laws at least locally.
\end{enumerate}

First, we introduce notation, choose parameters, and prove the simple lemma
which is important for further analysis. Suppose that we have exactly
$0\le\Sn\le n$ commuting Killing vector fields and not more, and variables are
nevertheless completely separated. We assume that variables corresponding to
Killing vectors are separated, and they are the first $\Sn$ coordinates. There
are only two possibilities left: diagonal components of separable metrics can
differ from zero or vanish, which is possible only for indefinite metrics.
Assume that the number of nonzero diagonal metric components is equal to
$0\le\Sm\le n$, and they precede zero components. We shall see in what follows
that the inequality
\begin{equation}                                                  \label{anbcpf}
  n-\Sn-\Sm\le\Sn\qquad\Leftrightarrow\qquad2\Sn+\Sm\le n
\end{equation}
must hold, otherwise the metric becomes degenerate.

A curve $x^\al(\tau)$, $\tau\in\MR$ is called isotropic if the tangent vector
$dx^\al$ is null, i.g.\ $g_{\al\al}\equiv0$. By analogy, we call the coordinate
line {\em coisotropic} if $g^{\al\al}\equiv0$. Sure, coisotropic lines exist
only on pseudoriemannian manifolds. Thus the last $n-\Sn-\Sm$ coordinates are
coisotropic. We shall see in what follows that the nonzero diagonal metric
components correspond to indecomposable quadratic conservation laws.

Introduce notation. Divide all coordinates on three groups
$(x^\al,y^\mu,z^\vf)\in\MM$, where indices from the beginning, middle, and end
of the Greek alphabet take the following values:
\begin{equation}                                                  \label{abfgdr}
 \begin{aligned}
  \al,\bt,\dotsc=&1,\dotsc,\Sn & & \text{(commuting Killing vectors)},
\\
  \mu,\nu,\dotsc=&\Sn+1,\dotsc,\Sn+\Sm & & (\text{quadratic conservation laws},
  ~g^{\mu\mu}\ne0),
\\
 \vf,\phi,\dotsc=&\Sn+\Sm+1,\dotsc,n & & (\text{coisotropic coordinates},
 ~g^{\vf\vf}\equiv0).
\end{aligned}
\end{equation}
In this way we described all possible types of separable metrics.

Suppose that all variables corresponding to Killing vector fields are separated,
and the first $\Sn$ coordinates are chosen cyclic. Then the corresponding
separating functions and conservation laws are
\begin{equation*}
  W'_\al=c_\al\qquad\text{and}\qquad p_\al=c_\al.
\end{equation*}

All parameters are also divided into three groups $(c,d,a)$. Parameters $c$ are
already introduced and parameters $d$ and $a$ are defined as follows. In a fixed
point $\Sp\in\MM$, we set
\begin{equation}                                                  \label{ifleki}
\begin{aligned}
  d_{ii}:=W^{\prime2}_\mu\big|_{\Sp,\mu=i}, & & i=\Sn+1,\dotsc,\Sn+\Sm,
\\
  a_r:=W'_\vf\big|_{\Sp,\vf=r},&& r=\Sn+\Sm+1,\dotsc,n.
\end{aligned}
\end{equation}
which is always possible. For convenience, we consider parameters $d_{ii}$ as
diagonal elements of matrix $d=(d_{ij})$,  $d_{ij}=0$ for $i\ne j$ (to preserve
the number of independent parameters). Sure, all diagonal elements are positive,
$d_{ii}>0$. Parameters $d$ for quadratic conservation laws and $a$ for
coisotropic coordinates are enumerated by Latin indices from the middle and end
of the alphabet, respectively.

At a fixed point $\Sp$, the Hamilton--Jacobi equation becomes
\begin{equation*}
  g^{\al\bt}_\Sp c_\al c_\bt+2g^{\al i}_\Sp c_\al\sqrt{d_{ii}}
  +2g^{\al r}_\Sp c_\al a_r+g^{ij}_\Sp\sqrt{d_{ii}d_{jj}}
  +2g^{ir}_\Sp\sqrt{d_{ii}}\,a_r+g^{rs}_\Sp a_r a_s=2E,
\end{equation*}
where $g^{rr}=0$. Therefore parameters $(c,d,a)$ and $E$ are simply related.
Choose the independent set of parameters in the following ways.

{\bf Case 1}. Energy $E$ corresponds to the coisotropic coordinate. Redefine
$a_n:=2E$. Then the set of independent parameters is
\begin{equation}                                                  \label{ndhtgr}
  (c_\al,d_{ij},a_{\Sn+\Sm+1},\dotsc,a_{n-1\,n-1},a_n:=2E).
\end{equation}

{\bf Case 2}. Energy $E$ corresponds to the indecomposable quadratic
conservation law.
Redefine $d_{\Sn+\Sm\,\Sn+\Sm}:=2E$. Then the set of independent parameters is
\begin{equation}                                                  \label{nsldkj}
  (c_\al,d_{\Sn+1\,\Sn+1},\dotsc,d_{\Sn+\Sm-1\,\Sn+\Sm-1},d_{\Sn+\Sm\,\Sn+\Sm}
  :=2E,a_r).
\end{equation}

This is always possible because redefinition of parameters $a_n$ and
$d_{\Sn+\Sm\,\Sn+\Sm}$ leads to multiplication of separating functions $W'_n$
and $W'_{\Sn+\Sm}$ on some nonzero constants.

We are looking for solutions of the Hamilton--Jacobi equation for separating
functions $W^{\prime2}_{\mu}$ and $W'_\vf$ within the class functions linear
in parameters $d$ and $a$:
\begin{align}                                                     \label{anbcfr}
  W^{\prime2}_\mu:=&b_{\mu\mu}{}^{ij}(y^\mu,c)d_{ij}
  +b_{\mu\mu}{}^r(y^\mu,c)a_r+k_{\mu\mu}(y^\mu,c)>0,
\\                                                                \label{nksjkh}
  W'_\vf:=&b_{\vf}{}^{ij}(z^\vf,c)d_{ij}+b_\vf{}^r(z^\vf,c)a_r+l_\vf(z^\vf,c),
\end{align}
where $b_{\mu\mu}{}^{ii}(y^\mu,c)$, $b_{\mu\mu}{}^r(y^\mu,c)$,
$b_\vf{}^{ii}(z^\vf,c)$, and $b_\vf{}^r(z^\vf,c)$ are some functions of single
coordinate and, possibly, the first group of parameters $c$. We suppose that
matrices $b_{\mu\mu}{}^{ii}$ whose elements are enumerated by pairs of indices
$\mu\mu$ and $ii$, and $b_\vf{}^r$ are nondegenerate. This choice is
justified by the result: the functional Hamilton--Jacobi equation has solution
within this class of separating functions, and our aim is to find at least one
solution. Then the Hamilton--Jacobi equation becomes
\begin{equation}                                                  \label{andjks}
  g^{\al\bt}c_\al c_\bt+2g^{\al\mu}c_\al W'_\mu+2g^{\al\vf}c_\al W'_\vf
  +g^{\mu\nu}W'_\mu W'_\nu+2g^{\mu\vf}W'_\mu W'_\vf+g^{\vf\phi}W'_\vf W'_\phi
  =2E.
\end{equation}
This equality must hold for all values of coordinates and parameters. For terms
quadratic in $a$, we have
\begin{equation*}
  g^{\vf\phi}b_\vf{}^r b_\phi{}^s a_ra_s\equiv0\qquad\Rightarrow\qquad
  g^{\vf\phi}\equiv0,
\end{equation*}
because matrix $b_\vf{}^r$ is nondegenerate. That is the whole metric block for
coisotropic coordinates must be equal to zero: $g^{\vf\phi}\equiv0$ for all
$\vf$ and $\phi$.

There are irrational function on $d$ on the left hand side of Eq.~(\ref{andjks})
which cannot be cancelled. Therefore the equalities $g^{\al\mu}\equiv0$ and
$g^{\mu\vf}\equiv0$ must hold. In addition block $g^{\mu\nu}$ must be diagonal.
Thus we have proved
\begin{lemma}                                                     \label{lnnsgf}
If independent parameters are chosen as in Eqs.~(\ref{ndhtgr}) or (\ref{nsldkj})
and separating functions have form (\ref{anbcfr}), (\ref{nksjkh}), then
separable metric has block form
\begin{equation}                                                  \label{aodewq}
  g^{**}=\begin{pmatrix} g^{\al\bt}(y,z) & 0 & g^{\al\phi}(y,z) \\
  0 & g^{\mu\nu}(y,z) & 0 \\ g^{\vf\bt}(y,z) & 0 & 0 \end{pmatrix},
\end{equation}
where block $g^{\mu\nu}$ is diagonal, and the star takes all values:
$*:=(\al,\mu,\vf)$.
\end{lemma}

In this way, three items of our strategy are met.

We shall see in what follows that variables are separated differently
depending on which group of parameters contains energy $E$: either $a_n=2E$
(case 1), or $d_{\Sn+\Sm\,\Sn+\Sm}=2E$ (case 2). Therefore each separable
metric belongs to one of the equivalence class $[\Sn,\Sm,n-\Sn-\Sm]_{1,2}$,
where the index shows the location of $E$, if it is considered as independent
parameter, i.e.\ if $\Sm\ne0$ and/or $n-\Sn-\Sm\ne0$.
All Riemannian positive definite metrics belong to classes $[\Sn,n-\Sn,0]_2$.

One and the same metric in different separable coordinates can be the member of
different classes. For example, the Euclidean metric on a plane in the Cartesian
and polar coordinates belong to classes $[2,0,0]$ and $[1,1,0]_2$, respectively
(see example \ref{esgwfr}).

In the Hamiltonian formalism, the complete variables separation leads to $n$
independent conservation laws in involution. For their derivation in explicit
form, we have to replace $W'_{\mu}\mapsto p_\mu$, $W'_r\mapsto p_r$ and solve
equalities (\ref{anbcfr}), (\ref{nksjkh}) with respect to parameters $d_{ii}$
and $a_r$. This introduces some restrictions on matrix $b$ which will be put
later.

Now we solve functional Hamilton--Jacobi equation with respect to separating
functions and metric components for different values of $\Sm$ and $\Sn$.
\subsection{Quadratic conservation laws}
Suppose that variables in the Hamilton--Jacobi equation are completely
separated, but Killing vectors and coisotropic coordinates are absent. Then
$\Sn=0$, $\Sm=n$, and the Hamilton--Jacobi equation is
\begin{equation}                                                  \label{anvbfl}
  g^{\mu\nu}W'_\mu W'_\nu=2E,
\end{equation}
where separating functions $W'_\mu(y^\mu,d)$ depend on single coordinate $y^\mu$
and, in a general case, on all independent parameters
\begin{equation}                                                  \label{ancbdo}
  d:=(d_{ij})=\diag(d_{11},\dotsc,d_{n-1\,n-1},d_{nn}:=2E).
\end{equation}
Due to lemma \ref{lnnsgf} the separable metric must be diagonal.
Separating functions have the form (\ref{anbcfr})
\begin{equation}                                                  \label{eodgfu}
  W^{\prime2}_\mu(y^\mu,d)=b_{\mu\mu}{}^{ij}(y^\mu) d_{ij}>0,
\end{equation}
where $b_{\mu\mu}{}^{ii}(y^\mu)$ is some invertible $(n\times n)$-matrix, whose
rows depend on single coordinate, and the inequalities restrict the form
of matrix $b$ for given parameters $d$. Without loss of generality, we assume
\begin{equation}                                                  \label{anbdgf}
  b_{\mu\nu}{}^{ij}\big|_{i\ne j}\equiv0,\qquad
  b_{\mu\nu}{}^{ij}\big|_{\mu\ne\nu}\equiv0,
\end{equation}
because the parameter matrix $d_{ij}$ and metric $g^{\mu\nu}$ are diagonal.

Differentiate the Hamilton--Jacobi equation (\ref{anvbfl}) with respect to all
parameters. As the result, we obtain the system of $n$ linear equations for $n$
diagonal metric components:
\begin{equation}                                                  \label{awlbfz}
\begin{split}
  \sum_{\mu=1}^n g^{\mu\mu}W'_\mu\frac{\pl W'_\mu}{\pl d_{ii}}=&0,\qquad
  i=1,\dotsc,n-1,
\\
  \sum_{\mu=1}^ng^{\mu\mu}W'_\mu\frac{\pl W'_\mu}{\pl E}=&1.
\end{split}
\end{equation}
The determinant of this system of equations on metric components differ from
zero:
\begin{equation*}
  \hphantom{\qquad\ns(\al)}
  \det\left(W'_\mu\frac{\pl W'_\mu}{\pl d_{ii}}\right)=W'_1\dotsc W'_n
  \det\left(\frac{\pl W'_\mu}{\pl d_{ii}}\right)\ne0,
  \qquad\ns(\mu),
\end{equation*}
where matrix indices are enumerated by pairs $\mu\mu$ and $ii$ due to
condition (\ref{indbfd}) which takes the form
\begin{equation}                                                  \label{aneiua}
  \det\left(\frac{\pl W'_\mu}{\pl d_{ii}}\right)\ne0.
\end{equation}
in our case. Consequently, the system of equations (\ref{awlbfz}) has unique
solution.

To find metric components, we substitute Eqs.~(\ref{eodgfu}) in the left hand
side of the Hamilton--Jacobi equation (\ref{anvbfl}):
\begin{equation}                                                  \label{egfdas}
  g^{\mu\nu}b_{\mu\nu}{}^{ij} d_{ij}=2E.
\end{equation}
Since $d_{nn}=2E$, it implies that diagonal metric components are
\begin{equation}                                                  \label{adfghf}
  g^{\mu\mu}=b_{nn}{}^{\mu\mu},
\end{equation}
where $b_{nn}{}^{\mu\mu}(y)$ is the last row of the inverse matrix
$b_{ii}{}^{\mu\mu}$:
\begin{equation*}
  \sum_{i=1}^n b_{\nu\nu}{}^{ii}b_{ii}{}^{\mu\mu}=
  b_{\nu\nu}{}^{ij}b_{ij}{}^{\mu\mu}=\dl^\mu_\nu,
  \qquad
  \sum_{\mu=1}^n b_{ii}{}^{\mu\mu} b_{\mu\mu}{}^{jj}=
  b_{ii}{}^{\mu\nu} b_{\mu\nu}{}^{jj}=\dl^j_i,
\end{equation*}
whose elements depend on all coordinates in general. In addition, we must
require that all elements of the last row of matrix (\ref{adfghf}) differ from
zero, otherwise the metric becomes degenerate.

Thus the problem is solved, and the complete integral of the Hamilton--Jacobi
equation is
\begin{equation*}
  W(y,d)=\sum_{\mu=1}^n W_\mu(y^\mu,d),
\end{equation*}
where the right hand side contains primitives
\begin{equation*}
  \hphantom{\qquad\qquad\ns(\al)}
  W_\mu(y^\mu,d):=\int\!\! dy^\mu\sqrt{b_{\mu\mu}{}^{ij}d_{ij}},
  \qquad\qquad\ns(\mu),\quad\forall\mu.
\end{equation*}
Integration constants in the last equalities are inessential because the action
function is defined up to a constant.

Equalities (\ref{eodgfu}) imply $n$ quadratic conservation laws
\begin{equation}                                                  \label{amnssg}
  b_{ii}{}^{\mu\nu}p_\mu p_\nu=d_{ii}
\end{equation}
in the Hamiltonian formulation. In general, the left hand side depends on all
coordinates and momenta. These conservation laws correspond to Killing tensors
of second rank.

Rewrite Eq.~(\ref{amnssg}) in the form
\begin{equation*}
  p_\mu^2=b_{\mu\mu}{}^{ij}d_{ij}.
\end{equation*}
It implies that quadratic conservation laws are indecomposable if and only if
each row of matrix $b_{\mu\mu}{}^{ii}$ contains at least two nonzero and
not proportional elements. Otherwise, the square root can be taken, and linear
conservation law appears, which contradicts the assumption on the absence of
Killing vectors.

Thus we have proved
\begin{theorem}                                                   \label{tgdjke}
If all diagonal inverse metric components differ from zero, $g^{\mu\mu}\ne0$,
and Killing vector fields and coisotropic components are absent, then there
exist such independent parameters (\ref{ancbdo}) and separating functions
(\ref{eodgfu}), that all conservation laws are quadratic, and the
Hamilton--Jacobi equation admits complete separation of variables if and only if
when inverse metric $g^{\mu\nu}$ is diagonal with components (\ref{adfghf}),
where $b_{ii}{}^{\mu\mu}$ is the inverse matrix to an arbitrary nondegenerate
matrix $b_{\mu\mu}{}^{ii}(y^\mu)$, whose rows depend on single coordinate
$y^\mu$ and contain at least two nonzero not proportional elements. In addition,
inequalities $b_{nn}{}^{\mu\mu}\ne0$ must hold for all values of index $\mu$,
and arbitrary functions $b_{\mu\mu}{}^{ii}$ must be chosen in such a way that
the system of equations for separating functions $W'_\mu$ (\ref{eodgfu}) have
real solutions for given values of parameters $d$.
\end{theorem}

Consequently, diagonal metric components are parameterized by $n^2$ arbitrary
functions on single coordinate $b_{\mu\mu}{}^{ii}(y^\mu)$, satisfying three
conditions: $\det b_{\mu\mu}{}^{ii}\ne0$, $b_{nn}{}^{\mu\mu}\ne0$ for all $\mu$,
and each row of matrix $b_{\mu\mu}{}^{ii}$ must contain at least two nonzero
not proportional elements. Note that components of the inverse matrix
$b_{ii}{}^{\mu\mu}(y)$ depend on all coordinates in general.

In fact, we prove the following. The right hand sides of Eqs.~(\ref{eodgfu})
contain arbitrary functions of single coordinates linear and homogeneous in $d$.
Therefore the separable metric is found for arbitrary indecomposable quadratic
conservation laws. All restrictions on arbitrary functions in the conservation
laws follow from nondegeneracy of the separable metric, indecomposability of the
quadratic Killing tensors, and existence of solutions for separating functions
$W'_\mu$ in Eqs.~(\ref{eodgfu}).

In the Hamiltonian formulation, we have $n$ involutive quadratic conservation
laws (\ref{amnssg}). Consequently, the Hamiltonian system is Liouville
integrable.

\begin{exa}                                                       \label{ecddsa}
Consider two dimensional Euclidean space $(y^1,y^2)\in\MR^2$  and choose
matrix $b$ in a general form
\begin{equation}                                                  \label{andbcf}
  b_{\mu\mu}{}^{ii}:=\begin{pmatrix} \phi_{11}(y^1) & \phi_{12}(y^1) \\
  \phi_{21}(y^2) & \phi_{22}(y^2)\end{pmatrix}
  \qquad\Rightarrow\qquad
  b_{ii}{}^{\mu\mu}:=\frac1{\det b}
  \begin{pmatrix} ~~\phi_{22} & -\phi_{12} \\ -\phi_{21} & ~~\phi_{11}
  \end{pmatrix}
\end{equation}
and assume that
\begin{equation}                                                  \label{eqhdpk}
  \det b=\phi_{11}\phi_{22}-\phi_{12}\phi_{21}\ne0.
\end{equation}
If all matrix elements differ from zero, then conditions of theorem \ref{tgdjke}
hold: elements of the first and second rows depend on $y^1$ and $y^2$,
respectively, and each row contains two nonzero elements assumed to be not
proportional. Separable metric corresponding to matrix $b$ us parameterized by
four functions of single coordinates:
\begin{equation}                                                  \label{eqkduy}
  g^{\mu\nu}=\frac1{\det b}
  \begin{pmatrix} -\phi_{21} & 0 \\ ~~0 & \phi_{11} \end{pmatrix}.
\end{equation}

Now we consider particular cases. Let $\phi_{11}\equiv1$ and
$\phi_{21}\equiv-1$. Then
\begin{equation*}
  b_{\mu\mu}{}^{ii}:=\begin{pmatrix} ~~1 & \phi_{12}(y^1) \\
  -1 & \phi_{22}(y^2)\end{pmatrix} \qquad\Rightarrow\qquad
  b_{ii}{}^{\mu\mu}:=\frac1{\phi_{12}+\phi_{22}}
  \begin{pmatrix} \phi_{22} & -\phi_{12} \\ 1 & 1\end{pmatrix}.
\end{equation*}
The respective inverse metric is conformally Euclidean:
\begin{equation}                                                  \label{abcvdv}
  g^{\mu\nu}=\frac1{\phi_{12}+\phi_{22}}\begin{pmatrix} 1 & 0 \\
  0 & 1 \end{pmatrix}.
\end{equation}
The Hamilton--Jacobi equation becomes
\begin{equation*}
 W^{\prime2}_1+W^{\prime2}_2=2E(\phi_{12}+\phi_{22}).
\end{equation*}
Complete separation of variables yields
\begin{equation*}
  W^{\prime2}_1=d_{11}+2E\phi_{12},\qquad W^{\prime2}_2=-d_{11}+2E\phi_{22},
\end{equation*}
where $d_{11}$ and $E$ are two independent parameters. The conservation laws are
quadratic
\begin{equation}                                                  \label{eshyts}
  \frac1{\phi_{12}+\phi_{22}}\big(\phi_{22}p_1^2-\phi_{12}p_2^2\big)=d_{11},
  \qquad \frac1{\phi_{12}+\phi_{22}}\big(p_1^2+p_2^2\big)=2E.
\end{equation}
In general, they are indecomposable for $d_{11}\ne0$, $E>0$ and nontrivial
functions $\phi_{12}$, $\phi_{22}$. The parameters domain for $d_{11}$, $E$ and
acceptable form of arbitrary functions are defined by the following
inequalities:
\begin{equation*}
  d_{11}+2E\phi_{12}>0,\qquad -d_{11}+2E\phi_{22}>0.
\end{equation*}
If, for example, $d_{11}>0$ and $E>0$, then
\begin{equation*}
  \phi_{12}>-\frac {d_{11}}{2E},\qquad\phi_{22}>\frac {d_{11}}{2E}.
\end{equation*}
Consequently this example is the particular two dimensional case of the
Liouville system for Riemannian metric considered in example \ref{ejsdfw}.

Now put $\phi_{11}\equiv-1$ and $\phi_{21}\equiv-1$. Then
\begin{equation*}
  b_{\mu\mu}{}^{ii}:=\begin{pmatrix} -1 & \phi_{12}(y^1) \\
  -1 & \phi_{22}(y^2)\end{pmatrix} \qquad\Rightarrow\qquad
  b_{ii}{}^{\mu\mu}:=\frac1{\phi_{12}-\phi_{22}}
  \begin{pmatrix} \phi_{22} & -\phi_{12} \\ 1 & -1\end{pmatrix}.
\end{equation*}
These matrices imply conformally Lorentzian metric
\begin{equation}                                                  \label{abchdv}
  g^{\mu\nu}=\frac1{\phi_{12}-\phi_{22}}\begin{pmatrix} 1 & ~~0 \\ 0 & -1
  \end{pmatrix}.
\end{equation}
The respective Hamilton--Jacobi equation is
\begin{equation*}
  W^{\prime2}_1-W^{\prime2}_2=2E(\phi_{12}-\phi_{22}),
\end{equation*}
and variables are separated:
\begin{equation*}
  W^{\prime2}_1=-d_{11}+2E\phi_{12},\qquad W^{\prime2}_2=-d_{11}+2E\phi_{22}.
\end{equation*}
They are indecomposable for $d_{11}\ne0$, $E\ne0$, and nonconstant functions
$\phi_{12}$, $\phi_{22}$. The conservation laws are quadratic:
\begin{equation}                                                  \label{ebnhbn}
  \frac1{\phi_{12}-\phi_{22}}\big(\phi_{22}p_1^2-\phi_{12}p_2^2\big)=d_{11},
  \qquad\frac1{\phi_{12}-\phi_{22}}\big(p_1^2-p_2^2\big)=2E.
\end{equation}
The parameter domain of definition and the form of arbitrary functions are
defined by inequalities:
\begin{equation*}
  -d_{11}+2E\phi_{12}>0,\qquad -d_{11}+2E\phi_{22}>0.
\end{equation*}
If $E>0$, then there are two possibilities when $\phi_{12}-\phi_{22}\ne0$:
\begin{equation*}
  \phi_{12}>\phi_{22}>\frac{d_{11}}{2E},\qquad
  \phi_{22}>\phi_{12}>\frac{d_{11}}{2E},
\end{equation*}
for $g^{11}>0$ and $g^{11}<0$, respectively. This is also two-dimensional
Liouville system considered in example \ref{ejsdfw} but for Lorentzian signature
metric.

Let $\phi_{12}\equiv0$ and $\phi_{21}\equiv-1$. Then
\begin{equation*}
  b_{\mu\mu}{}^{ii}:=\begin{pmatrix} ~~\phi_{11}(x) & 0 \\
  -1 & \phi_{22}(y)\end{pmatrix} \qquad\Rightarrow\qquad
  b_{ii}{}^{\mu\mu}:=\frac1{\phi_{11}\phi_{22}}
  \begin{pmatrix} \phi_{22} & 0 \\ 1 & \phi_{11}\end{pmatrix}.
\end{equation*}
The conditions of theorem \ref{tgdjke} are not fulfilled because the first row
of matrix $b_{\mu\mu}{}^{ii}$ contains only one nonzero element. The respective
inverse metric is
\begin{equation}                                                  \label{abkvdv}
  g^{\mu\nu}=\frac1{\phi_{11}\phi_{22}}\begin{pmatrix} 1 & 0 \\
  0 & \phi_{11} \end{pmatrix}.
\end{equation}
The Hamilton--Jacobi equation for this metric
\begin{equation*}
  W^{\prime2}_1+\phi_{11}W^{\prime2}_2=2E\phi_{11}\phi_{22},
\end{equation*}
admits complete separation of variables:
\begin{equation*}
  W^{\prime2}_1=\phi_{11}d_{11},\qquad W^{\prime2}_2=-d_{11}+2E\phi_{22}.
\end{equation*}
The conservation laws are quadratic:
\begin{equation}                                                  \label{efwtsy}
  \frac1{\phi_{11}\phi_{22}}p_1^2=d_{11},\qquad
  \frac1{\phi_{11}\phi_{22}}\big(p_1^2+\phi_{11} p_2^2\big)=2E.
\end{equation}
The parameter domain of definition and admissible form of arbitrary functions
are defined by inequalities:
\begin{equation*}
  \phi_{11}d_{11}>0,\qquad -d_{11}+2E\phi_{22}>0.
\end{equation*}
For $d_{11}>0$ and $E>0$, for example, they restrict arbitrary functions:
\begin{equation*}
  \phi_{11}>0,\qquad \phi_{22}>\frac{d_{11}}{2E}.
\end{equation*}
In this case the first conservation law (\ref{efwtsy}) is decomposable: it is
the square of linear conservation law
$p_1/\sqrt{\phi_{11}\phi_{22}}=\sqrt{d_{11}}$. Consequently,  we have in fact
one linear and one quadratic conservation law.
\qed\end{exa}

Now we simplify the form of matrix $b_{\mu\mu}{}^{ii}$ using the canonical
transformation. It turns out that one of nonzero elements in each row of
matrix $b_{\mu\mu}{}^{ii}$ can be transformed to unity. For example, let
$b_{\mu\mu}{}^{ii}\ne0$ for given $\mu$ and $i$. For definiteness we assume that
it is the diagonal element, $\mu=i$. Choose the generating function of canonical transformation
for one pair of canonically conjugate variables $(y^\mu,p_\mu)\mapsto(Y^i,P_i)$
as
\begin{equation}                                                  \label{abdvfi}
  \hphantom{\qquad\qquad\ns(\al)}
  S_2(y,P):=\int\!\!dy^\mu\sqrt{|b_{\mu\mu}{}^{ii}|}\,P_i,
  \qquad\qquad\ns(\mu,i).
\end{equation}
Note that it must be linear in $P$, because otherwise the coordinate
transformation $x\mapsto X(x)$ depends on momenta which contradicts the
equivalence relation. Then we get for fixed $\mu$:
\begin{equation*}
  p_\mu=\frac{\pl S_2}{\pl y^\mu}=\sqrt{|b_{\mu\mu}{}^{ii}|}\,P_i,
  \qquad Y^i=\frac{\pl S_2}{\pl P_i}=\int\!\!dy^\mu
  \sqrt{|b_{\mu\mu}{}^{ii}|}, \qquad\qquad\ns(\mu,i),
\end{equation*}
and Eq.~(\ref{eodgfu}) becomes
\begin{equation*}
  |b_{\mu\mu}{}^{ii}|\,\tilde W^{\prime2}_i=b_{\mu\mu}{}^{ii}\,d_{ii}
  +\sum_{j\ne i}b_{\mu\mu}{}^{jj}\,d_{jj}
\end{equation*}
for transformed separating function $\tilde W_i$. Dividing it by
$|b_{\mu\mu}{}^{ii}|$, we obtain the quadratic equality
\begin{equation}                                                  \label{anfbgr}
  \tilde W^{\prime2}_i=\pm d_{ii}+\sum_{j\ne i}\tilde b_{ii}{}^{jj}\,d_{jj}
\end{equation}
with some new functions $\tilde b_{ii}{}^{jj}(Y^i)$. Similar
transformation can be performed for each row. So, without loss of generality,
one of nonzero elements in each row can be set to unity because signs $\pm1$
can be attributed to parameters $d$. It means that canonical separable metric is
parameterized by $n^2-n$ functions of single coordinates.

In fact, this statement is evident. Indeed, Eq.~(\ref{eodgfu}) depends only on
single coordinate, and one of nonzero elements can be transformed to $\pm1$ by
coordinate transformation  $y^\mu\mapsto\tilde y^\mu(y^\mu)$.
\subsection{Linear and quadratic conservation laws}               \label{sbbcdh}
Assume that metric admits exactly $1\le\Sn<n$ and not more commuting Killing
vector fields, and coisotropic coordinates are absent, $n-\Sn-\Sm=0$. There is
one linear conservation law for each Killing vector. Then we need additional
$n-\Sn$ independent involutive conservation laws to provide complete
integrability. These conservation laws are quadratic as was shown in the
preceeding section. Now we have two groups of coordinates $(x^\al,y^\mu)\in\MM$
and two groups of independent parameters $c$ and $d$ (\ref{nsldkj}). Separable
metric (\ref{aodewq}) in this case is block diagonal
\begin{equation}                                                  \label{ahdjsk}
  g^{**}=\begin{pmatrix} g^{\al\bt} & 0 \\ 0 & g^{\mu\nu} \end{pmatrix},
\end{equation}
where the lower block $g^{\mu\nu}$ is diagonal.

Functions $k$ in equality (\ref{anbcfr}) must be quadratic in $c$ as the
consequence of the Hamilton--Jacobi equation. Therefore separating functions are
\begin{equation}                                                  \label{nsldjj}
  W^{\prime2}_\mu=b_{\mu\mu}{}^{ij}d_{ij}+k^{\al\bt}_{\mu\mu}c_\al c_\bt>0,
\end{equation}
where $b_{\mu\mu}{}^{ii}(y^\mu)$ is some nondegenerate matrix whose elements in
each row depend on single coordinate corresponding to the number of the row, and
functions $k^{\al\bt}_{\mu\mu}(y^\mu)$ are arbitrary and symmetric in indices
$\al$ and $\bt$ but can be nondiagonal in them. Their indices can be written one
under another because they will never be lowered or raised.

There is the new property. We shall see that the term
$k^{\al\bt}_{\mu\mu}c_\al c_\bt$ in Eq.~(\ref{nsldjj}) must not vanish for
nondegenerate metric. Therefore the requirement that each row of matrix
$b_{\mu\mu}{}^{ii}$ contains at least two not proportional elements is
unnecessary.

After separating the first group of coordinates, the Hamilton--Jacobi equation
becomes
\begin{equation}                                                  \label{egdfrj}
  g^{\al\bt}c_\al c_\bt+g^{\mu\nu}(b_{\mu\nu}{}^{ij}d_{\ij}
  +k^{\al\bt}_{\mu\nu}c_\al c_\bt)=2E,
\end{equation}
where matrix $g^{\mu\nu}$ is diagonal. We get previous expression (\ref{adfghf})
for elements in the second block because $d_{nn}:=2E$. Therefore
\begin{equation}                                                  \label{egdfre}
  g^{\al\bt}c_\al c_\bt+g^{\mu\nu}k^{\al\bt}_{\mu\nu}c_\al c_\bt=0.
\end{equation}
This equality must be fulfilled for all $c$, and defines the upper block of the
inverse separable metric (\ref{ahdjsk}):
\begin{equation}                                                  \label{anfbdj}
 g^{\al\bt}=-k^{\al\bt}_{\mu\nu}g^{\mu\nu},
\end{equation}
where $k^{\al\bt}_{\mu\mu}(y^\mu)$ are arbitrary functions on single
coordinates providing nondegeneracy of $g^{\al\bt}$ and positivity of right hand
sides of Eqs.~(\ref{nsldjj}). Thus we have proved
\begin{theorem}                                                   \label{theojk}
If separable metric admits exactly $1\le\Sn<n$ and not more commuting Killing
vector fields, and coisotropic coordinates are absent, $n-\Sn-\Sm=0$, then there
exists such set of independent parameters (\ref{nsldkj}) and separating
functions (\ref{nsldjj}), that separable inverse metric has block form
(\ref{ahdjsk}). The lower block is diagonal with elements (\ref{adfghf}), where
$b_{ii}{}^{\mu\mu}(y)$ is the matrix inverse to matrix
$b_{\mu\mu}{}^{ii}(y^\mu)$ whose elements in each row depend only on single
coordinates, and all elements of the last row differ from zero,
$b_{nn}{}^{\mu\mu}\ne0$. The upper block  has form (\ref{anfbdj}) with arbitrary
functions $k^{\al\bt}_{\mu\mu}(y^\mu)=k^{\bt\al}_{\mu\mu}(y^\mu)$ depending on
single coordinates. In addition, the upper block (\ref{egdfre}) must be
nondegenerate and provide positiveness of right hand sides of
Eqs.~(\ref{nsldjj}). In the Hamiltonian formulation, there are $n-\Sn$
indecomposable quadratic conservation laws:
\begin{equation}                                                  \label{abcvdk}
\begin{split}
  p_\al=&c_\al,
\\
  b_{ii}{}^{\mu\nu}\big(p_\mu p_\nu-k^{\al\bt}_{\mu\nu}c_\al c_\bt\big)
  =&d_{ii}.
\end{split}
\end{equation}
\end{theorem}
Quadratic conservation laws (\ref{abcvdk}) appear after contraction of
Eqs.~(\ref{nsldjj}) with inverse matrix $b_{ii}{}^{\mu\mu}$.

Thus the canonical separable inverse metric of type $[\Sn,n-\Sn,0]_2$ has block
diagonal form
\begin{equation}                                                  \label{anncbv}
  g^{**}=\begin{pmatrix} -k^{\al\bt}_{\mu\nu}g^{\mu\nu} & 0 \\
  0 & g^{\mu\nu}\end{pmatrix},
\end{equation}
where the lower block $g^{\mu\nu}$ is diagonal (\ref{adfghf}), and stars denote
all indices $*=(\al,\mu)$.

We can also simplify matrix $b_{\mu\mu}{}^{ii}$ using canonical transformation
(\ref{abdvfi}) from previous section. Therefore one of nonzero elements in each
row of matrix $b$ can se set to unity without loss of generality.

\begin{exa}\rm                                                       \label{esjjhg}
Consider the simplest example, when separable metric admits only one
indecomposable quadratic conservation law, i.e.\ $\Sn=n-1$ and $d_{nn}=2E$. In
this case, there is only one coordinate $y$, which may not be enumerated. Matrix
$b_{\mu\mu}{}^{ii}$ consists of only one element $b(y)\ne0$, and its inverse is
$1/b$. Due to theorem \ref{theojk}, the most general canonical separable metric
(\ref{anncbv}) is block diagonal:
\begin{equation}                                                  \label{ancbda}
  g^{**}=
  \begin{pmatrix} g^{\al\bt}(y) & 0 \\ 0 & \frac1{b(y)} \end{pmatrix}
  =\begin{pmatrix} -k^{\al\bt}(y)\frac1{b(y)} & 0 \\ 0 & \frac1{b(y)}
  \end{pmatrix},\qquad\al,\bt=1,\dotsc,n-1,
\end{equation}
where the $(n-1)\times(n-1)$ block $g^{\al\bt}(y)$ can be arbitrary
nondegenerate matrix if we include arbitrary function $b(y)\ne0$ in the
definition of $k^{\al\bt}(y)$. Variables are separated as
\begin{equation*}
  W'_\al=c_\al,\qquad W^{\prime2}_n=2Eb+k^{\al\bt}c_\al c_\bt>0,
\end{equation*}
where the inequality restricts functions $b$ and $k$ for fixed $E$ and $c$.
The respective conservation laws are
\begin{equation}                                                  \label{abdvid}
  p_\al=c_\al,\qquad \frac1{b}\big(p_n^2-k^{\al\bt}c_\al c_\bt\big)=2E.
\end{equation}

To simplify metric (\ref{ancbda})  we perform the canonical transformation
$(y,p_n)\mapsto(Y,P_n)$ with the generating function
\begin{equation*}
  S_2:=\int\!\!dy\sqrt{|b(y)|}P_n,
\end{equation*}
leaving remaining variables $x^\al,p_\al$ untouched. Then
\begin{equation*}
  p_n=\frac{\pl S_2}{\pl y}=\sqrt{|b|}P_n,\qquad
  Y=\frac{\pl S_2}{\pl P_n}=\int\!\!dy\sqrt{|b|},
\end{equation*}
and the quadratic conservation law takes the form
\begin{equation*}
  \pm P^2_n+g^{\al\bt}c_\al c_\bt=2E,
\end{equation*}
After this transformation of variables, the canonical separable metric becomes
\begin{equation}                                                  \label{aghsjy}
  g^{**}=\begin{pmatrix} g^{\al\bt}(y) & 0 \\ 0 & \pm1\end{pmatrix},
\end{equation}
where the sign choice depends on the signature of the metric.
\qed\end{exa}
\subsection{Coisotropic coordinates}                              \label{sgsrsh}
We showed in section \ref{sbskdy} that linear conservation laws for the
``symmetric'' choice of independent parameters correspond to Killing vector
fields. However, if energy $E$ in the right hand side of the Hamilton--Jacobi
equation is considered as independent parameter, linear conservation laws not
related to Killing vectors may appear. This possibility arises only for
indefinite metrics having coisotropic coordinates. For Riemannian positive
definite metrics linear conservation laws are always related to Killing vector
fields.

Existence of these linear conservation laws can take place only when sufficient
number of commuting Killing vectors are simultaneously present. Suppose that
$\Sm=0$, i.e.\ we have only Killing vector fields and coisotropic coordinates.
then we have two groups of coordinates $(x^\al,z^\vf)\in\MM$, where indices
take the following values
\begin{equation}                                                  \label{ibdfgf}
  \al,\bt,\dotsc=1,\dotsc,\Sn;\qquad\vf,\phi,\dotsc=\Sn+1,\dotsc,n;
  \qquad \frac n2\le\Sn<n.
\end{equation}
By assumption, there are independent parameters
\begin{equation}                                                  \label{nksjko}
  \lbrace c_1,\dotsc,c_\Sn,a_{\Sn+1},\dotsc,a_{n-1},a_n:=2E\rbrace.
\end{equation}

Separable metric of type $[\Sn,0,n-\Sn]_1$ has block form (\ref{aodewq})
\begin{equation}                                                  \label{andbfw}
  g^{**}=\begin{pmatrix} g^{\al\bt}(z) & g^{\al\phi}(z) \\
  g^{\vf\bt}(z) & 0 \end{pmatrix}.
\end{equation}
This metric must be nondegenerate, therefore the number of Killing vectors
cannot be less then the number of coisotropic coordinates
\begin{equation}                                                  \label{andhgf}
  \Sn\ge n-\Sn\qquad\Leftrightarrow\qquad\Sn\ge\frac n2,
\end{equation}
which was assumed at the very beginning (\ref{ibdfgf}). In addition, the
rank of rectangular matrix $g^{\al\phi}$ is equal to $n-\Sn$, otherwise
separable metric is degenerate.

After separation of cyclic coordinates condition (\ref{ijdhgf}) becomes
\begin{equation}                                                  \label{anfgth}
  \det\big(\pl^rW'_\vf\big)\ne0,
\end{equation}
and Hamilton--Jacobi equation (\ref{abdnfh}) for metric (\ref{andbfw}) takes the
form
\begin{equation}                                                  \label{aocvdj}
  g^{\al\bt}c_\al c_\bt+2g^{\al\vf}c_\al W'_\vf=2E.
\end{equation}
Separating functions $W'_\vf$ (\ref{nksjkh}) are linear in parameters $a$:
\begin{equation}                                                  \label{eqghst}
  W'_\vf=b_\vf{}^r(z^\vf,c)a_r+l_\vf(z^\vf,c),
\end{equation}
where $b_\vf{}^r$ is some nondegenerate matrix, whose elements in each row
depend only on single coordinate $z^\vf$ and, possibly, on the first group
of parameters $(c_\al)$, and $l_\vf$ are some functions also on single
coordinate and first group of parameters. Equating terms with $E$ in
Eq.~(\ref{aocvdj}), we obtain
\begin{equation}                                                  \label{anbfgf}
  2g^{\al\vf}c_\al=b_n{}^\vf,
\end{equation}
where $b_n{}^\vf$ is the last row of matrix $b_r{}^\vf$, which is inverse to
$b_\vf{}^r$: $b_r{}^\vf b_\phi{}^r=\dl_\phi^\vf$. This equality implies that
the last row of matrix $b_n{}^\vf(z,c)$, whose elements depend in general on all
coordinates $z$, are linear in $c_\al$. Let
\begin{equation}                                                  \label{ekjiuf}
  b_n{}^\vf:=b_n{}^{\al\vf}c_\al,
\end{equation}
where $b_n{}^{\al\vf}(z)$ is a set of arbitrary functions on $z$, such that
matrix $(b_r{}^{\al\vf}c_\al)$ is not degenerate. Equality (\ref{anbfgf})
must hold for all values of $c_\al$, therefore
\begin{equation}                                                  \label{anbdgf}
  2g^{\al\vf}=b_n{}^{\al\vf}.
\end{equation}
Now we substitute the expression for separating functions $W'_\vf$
(\ref{eqghst}) and use Eq.~(\ref{anbfgf}) in the Hamilton--Jacobi equation
(\ref{aocvdj}):
\begin{equation}                                                  \label{amdggf}
  g^{\al\bt}c_\al c_\bt+b_n{}^{\al\vf} c_\al l_\vf=0.
\end{equation}
It implies that functions $l_\vf$ are linear and homogeneous in $c_\al$:
\begin{equation}                                                  \label{ejdhht}
  l_\vf=l^\al_\vf(z^\vf)c_\al,
\end{equation}
where $l^\al_\vf(z^\vf)$ are some functions depending on single coordinate.
Indices here can be written one over the other because they will be never raised
or lowered. Then Eq.~(\ref{amdggf}) defines the square block for Killing vectors
\begin{equation}                                                  \label{abdfgf}
  2g^{\al\bt}=-b_n{}^{\al\vf}l^\bt_\vf-b_n{}^{\bt\vf}l^\al_\vf.
\end{equation}

Thus Hamilton--Jacobi equation (\ref{aocvdj}) is solved, and the separable
metric is
\begin{equation}                                                  \label{invbfg}
  g^{**}=\frac12\begin{pmatrix}
  -b_n{}^{\al\chi}l_\chi^\bt-b_n{}^{\bt\chi}l_\chi^\al & b_n{}^{\al\phi} \\[4pt]
  b_n{}^{\bt\vf} & 0 \end{pmatrix},
\end{equation}
where the star denotes all indices, $*=(\al,\vf)$.

To simplify the separable metric, we perform the canonical transformation with
generating function
\begin{equation}                                                  \label{amgnfh}
  S_2:=x^\al P_\al+z^\vf P_\vf+\sum_{\vf=\Sn+1}^n\int\!\!dz^\vf
  l^\al_\vf(z^\vf)P_\al.
\end{equation}
Then
\begin{equation*}
\begin{aligned}
  p_\al=&P_\al, & \quad p_\vf=&P_\vf+l^\al_\vf P_\al,
\\
  X^\al=&x^\al+\sum_{\vf=\Sn+1}^n\int\!\!dz^\vf\,l^\al_\vf, &
  Z^\phi=&z^\vf,\qquad\ns(\vf).
\end{aligned}
\end{equation*}
The last two equalities define the coordinate transformation with the Jacobi
matrix
\begin{equation*}
  \frac{\pl(X^\bt,Z^\phi)}{\pl(x^\al,z^\vf)}=
  \begin{pmatrix} \dl_\al^\bt & 0 \\ l^\bt_\vf & \dl_\vf{}^\phi \end{pmatrix}.
\end{equation*}
The metric transforms as follows $g^{**}\mapsto\tilde g^{**}$ with
\begin{equation*}
  \tilde g^{\al\bt}=0,\qquad\tilde g^{\al\phi}=g^{\al\phi},
  \qquad\tilde g^{\vf\phi}=0,
\end{equation*}
where equality $g^{\vf\phi}\equiv0$ was used. After this transformation the
separable metric is simplified
\begin{equation}                                                  \label{invbff}
  g^{**}=\frac12\begin{pmatrix} 0 & b_n{}^{\al\phi} \\[4pt]
  b_n{}^{\bt\vf} & 0 \end{pmatrix}.
\end{equation}
This metric is always degenerate except for $n=2\Sn$. Then
\begin{equation*}
  \det g^{**}=\frac{(-1)^\Sn}{2^n}\det{}^2(b_n{}^{\al\phi})\ne0.
\end{equation*}

Now we are left with the problem to find such matrix $b_\vf{}^r(z,c)$ in
Eq.~(\ref{eqghst}), that equality (\ref{anbdgf}) be satisfied for all
coordinates $z$ and parameters $c$. In other words we have to extract explicitly
the dependence of matrix elements $b_\vf{}^r$ on parameters, because
Eq.~(\ref{ekjiuf}) contains elements of the inverse metric.
\begin{prop}
Matrix elements $(b_\vf{}^r)$ in Eq.~(\ref{eqghst}) satisfying equality
(\ref{anbdgf}) must have the form
\begin{equation}                                                  \label{isndgf}
  \hphantom{\qquad\qquad\ns(\al)}
  b_\vf{}^r(z^\vf,c)=\frac{\phi_\vf{}^r}{h_\vf^\al c_\al},
  \qquad\qquad\ns(\vf),
\end{equation}
where $\phi_\vf{}^r(z^\vf)$ is the nondegenerate matrix, whose elements of each
row depend on single coordinate and contain unity, and $h_\vf^\al(z^\vf)$ is a
set of arbitrary nonzero functions depending on single coordinates.
\end{prop}
\begin{proof}
Parameterise matrix $(b_\vf{}^r)$ as
\begin{equation*}
  \hphantom{\qquad\qquad\ns(\al)}
  b_\vf{}^r(z^\vf,c):=\frac{\phi_\vf{}^r(z^\vf,c)}{h_\vf(x^\vf,c)},
  \qquad\qquad\ns(\vf),
\end{equation*}
where $h_\vf$ are some nonzero functions, and
\begin{equation}                                                  \label{irkrjg}
  (\phi_\vf{}^r):=\begin{pmatrix} 1 & \phi_{\Sn+1}{}^{\Sn+2} &
  \phi_{\Sn+1}{}^{\Sn+3} & \cdots & \phi_{\Sn+1}{}^n \\[8pt]
  \phi_{\Sn+2}{}^{\Sn+1} & 1 & \phi_{\Sn+2}{}^{\Sn+3} & \cdots &
  \phi_{\Sn+2}{}^n \\[8pt]
  \phi_{\Sn+3}{}^{\Sn+1} & \phi_{\Sn+3}{}^{\Sn+2} & 1 & \cdots &
  \phi_{\Sn+3}{}^n \\
  \vdots & \vdots & \vdots & \ddots & \vdots \\
  \phi_n{}^{\Sn+1} & \phi_n{}^{\Sn+2} & \phi_n{}^{\Sn+3} & \cdots & 1
\end{pmatrix}.
\end{equation}
This matrix has unities on the diagonal, and each row of matrix $(b_\vf{}^r)$
is the quotient of the row $(\phi_\vf{}^r)$ by $h_\vf$.

Multiply Eq.~(\ref{anbdgf}) by matrix $b_\vf{}^r$ and sum over $\vf$:
\begin{equation*}
  g^{\al\vf}c_\al b_\vf{}^r=\sum_{\vf=\Sn+1}^n g^{\al\vf}c_\al
  \frac{\phi_\vf{}^r}{h_\vf}=\dl_n^r.
\end{equation*}
This equality must be fulfilled for all coordinates and parameters. Therefore
\begin{equation*}
  \phi_\vf{}^r(z^\vf,c)=\phi_\vf{}^r(z^\vf),\qquad
  h_\vf(x^\vf,c)=h_\vf^\al(z^\vf)c_\al,
\end{equation*}
where $h_\vf^\al(z^\vf)$ are some functions on single coordinate. It means that
matrix (\ref{irkrjg}) must not depend on $c$, and functions $h_\vf$ be linear
in $c$.
\end{proof}
It implies that matrix $(b_\vf{}^r)$ and consequently canonical separable metric
is parameterized by $2\Sn^2-\Sn$ arbitrary functions, and
\begin{equation}                                                  \label{anbdfr}
  \hphantom{\qquad\qquad\ns(\al)}
  2g^{\al\vf}=b_n{}^{\al\vf}=\phi_n{}^\vf h_\vf^\al,
  \qquad\qquad\ns(\vf),
\end{equation}
where $(\phi_r{}^\vf)$ is the matrix inverse to $(\phi_\vf{}^r)$. Thus we have
found canonical separable metric in the case of only Killing vector fields and
coisotropic coordinates.

\begin{theorem}                                                   \label{tdhfgd}
If separable metric admit exactly $\Sn$ and not more commuting Killing vector
fields, and all other coordinates are coisotropic, then there exists such set
of parameters (\ref{nksjko}) and separating functions (\ref{eqghst}), that the
Hamilton--Jacobi equation admit complete separation of variables if and only if
the dimension of the manifold is even, $n=2\Sn$, and canonical separable metric
has block form (\ref{invbff}). The off diagonal blocks are given by
Eqs.~(\ref{anbdfr}), where $(\phi_n{}^\vf)$ is the last row of the matrix
inverse to arbitrary nondegenerate matrix (\ref{irkrjg}), whose rows depend on
single coordinates, and the diagonal consists of unities. All conservation laws
are linear:
\begin{equation}                                                  \label{abndjh}
\begin{split}
  p_\al=&c_\al,
\\
  \sum_\vf\phi_r{}^\vf h_\vf^\al c_\al p_\vf=&a_r.
\end{split}
\end{equation}
\end{theorem}

The second conservation law (\ref{abndjh}) is linear in momenta only after
separation of cyclic coordinates. If this is not done, then it is quadratic
\begin{equation*}
  \sum_\vf\phi_r{}^\vf h_\vf^\al p_\al p_\vf=a_r.
\end{equation*}
The respective Killing tensor is indecomposable in general.

The Hamilton--Jacobi equation for canonical separable metric (\ref{invbff}) is
\begin{equation}                                                  \label{ifkdiu}
  \sum_\vf\phi_n{}^\vf h_\vf^\al W'_\al W'_\vf=2E,
\end{equation}
and variables are completely separated
\begin{equation*}
  \hphantom{\qquad\qquad\ns(\al)}
  W'_\al=c_\al,\qquad W'_\vf=\frac{\phi_\vf{}^r}{h_\vf^\al c_\al}a_r,
  \qquad\qquad\ns(\vf).
\end{equation*}
The unusual feature of this separation is that parameters $c$ appear in the
denominator.

Finally, using canonical transformations we simplify the form of the canonical
separable metric by setting one of the nonzero elements in each row of matrix
$h_\vf^\al$ to unity.
\subsection{General separation of variables}                      \label{sbdnsg}
If separable metric admits simultaneously commuting Killing vector fields,
indecomposable quadratic conservation laws, and coisotropic coordinates, then
coordinates are divided into three groups $(x,y,z)\in\MM$ described in section
\ref{sajwuy}. There are two possibilities for full sets of independent
parameters: (\ref{ndhtgr}) and (\ref{nsldkj}), and separable metric must have
block form (\ref{aodewq}). Separable functions are given by Eqs.~(\ref{anbcfr})
and (\ref{nksjkh}).

First, we consider case 2, when energy enters parameters for quadratic
conservation laws, $d_{\Sn+\Sm\,\Sn+\Sm}:=2E$. Without loss of generality, we
set
\begin{equation*}
  b_{\mu\nu}{}^{ij}\big|_{i\ne j}\equiv0,\qquad
  b_{\mu\nu}{}^{ij}\big|_{\mu\ne\nu}\equiv0,\qquad
  b_{\mu\nu}{}^r\big|_{\mu\ne\nu}\equiv0,\qquad
  k_{\mu\nu}\big|_{\mu\ne\nu}\equiv0,
\end{equation*}
because metric $g^{\mu\nu}$ and matrix of parameters $d_{ij}$ are diagonal.
The respective Hamilton--Jacobi equation is
\begin{equation}                                                  \label{anfkes}
  g^{\al\bt}c_\al c_\bt+g^{\mu\nu}\big(b_{\mu\nu}{}^{ij}d_{ij}+b_{\mu\nu}{}^ra_r
  +k_{\mu\nu}\big)
  +2g^{\al\phi}c_\al\big(b_\phi{}^{ij}d_{ij}+b_\phi{}^ra_r+l_\phi\big)=2E.
\end{equation}
Differentiate this equation consequently with respect to parameters $d$ and $a$:
\begin{equation}                                                  \label{ancbdf}
\begin{aligned}
  g^{\mu\nu}\frac{\pl(W'_\mu W'_\nu)}{\pl d_{ii}}
  +2g^{\al\vf}c_\al\frac{\pl W'_\vf}{\pl d_{ii}}=&0,\qquad i\ne\Sn+\Sm,
\\[4pt]
  g^{\mu\nu}\frac{\pl (W'_\mu W'_\nu)}{\pl d_{ii}}
  +2g^{\al\vf}c_\al\frac{\pl W'_\vf}{\pl d_{ii}}=&1,\qquad i=\Sn+\Sm,
\\[4pt]
  g^{\mu\nu}\frac{\pl(W'_\mu \tilde W'_\nu)}{\pl a_r}
  +2g^{\al\vf}c_\al\frac{\pl W'_\vf}{\pl a_r}=&0,
\end{aligned}
\end{equation}

This system of equation is considered as the system of linear algebraic
equations for $g^{\mu\nu}$ and linear combinations $g^{\al\vf}c_\al$. Its
determinant differs from zero
\begin{equation*}
  \det \begin{pmatrix}
    \displaystyle\frac{\pl(W^{\prime2}_\mu)}{\pl d_{ii}} &
    \displaystyle \frac{\pl W'_\vf}{\pl d_{ii}}
    \\[8pt] \displaystyle\frac{\pl(W^{\prime2}_\mu)}{\pl a_r} &
    \displaystyle \frac{\pl W'_\vf}{\pl a_r}  \end{pmatrix}
  = 2^{\Sm}\tilde W'_{\Sn+1}\dotsc \tilde W'_{\Sn+\Sm}\det \begin{pmatrix}
    \displaystyle\frac{\pl(W^{\prime}_\mu)}{\pl d_{ii}} &
    \displaystyle \frac{\pl W'_\vf}{\pl d_{ii}}
    \\[8pt] \displaystyle\frac{\pl(W^{\prime}_\mu)}{\pl a_r} &
    \displaystyle \frac{\pl W'_\vf}{\pl a_r}  \end{pmatrix} \ne0
\end{equation*}
due to condition (\ref{ijdhgf}).

Consider $(n-\Sn)\!\times\!(n-\Sn)$ matrix
\begin{equation}                                                  \label{andkfj}
  B:=\begin{pmatrix} b_{\mu\mu}{}^{ii}(y^\mu,c) & b_{\mu\mu}{}^r(y^\mu,c) \\
  b_\vf{}^{ii}(z^\vf,c) & b_\vf{}^r(z^\vf,c) \end{pmatrix},
\end{equation}
whose elements of each row depend on single coordinate and, possibly, the first
group of parameters $c$. It must be nondegenerate, as we shall see. The inverse
metric is
\begin{equation*}
  B^{-1}=\begin{pmatrix} b_{ii}{}^{\mu\mu} & b_{ii}{}^\vf \\
  b_r{}^{\nu\nu} & b_r{}^\vf \end{pmatrix},
\end{equation*}
where
\begin{alignat*}{2}
  b_{\mu\mu}{}^{ij}b_{ij}{}^{\nu\nu}+b_{\mu\mu}{}^rb_r{}^{\nu\nu}=&\dl_\mu^\nu,&
  \qquad b_{ii}{}^{\mu\nu}b_{\mu\nu}{}^{jj}+b_{ii}{}^\vf b_\vf{}^{jj}=&\dl_i^j,
\\
  b_{\mu\mu}{}^{ij}b_{ij}{}^\phi+b_{\mu\mu}{}^rb_r{}^\phi=&0, &
  \qquad\qquad b_{ii}{}^{\mu\nu}b_{\mu\nu}{}^r+b_{ii}{}^\vf b_\vf{}^r=&0,
\\
  b_\vf{}^{ij}b_{ij}{}^{\nu\nu}+b_\vf{}^rb_r{}^{\nu\nu}=&0, &
  b_r{}^{\mu\nu}b_{\mu\nu}{}^{jj}+b_r{}^\vf b_\vf{}^{jj}=&0,
\\
  b_\vf{}^{ij}b_{ij}{}^\phi+b_\vf{}^rb_r{}^\phi=&\dl_\vf^\phi, &
  b_r{}^{\mu\nu}b_{\mu\nu}{}^s+b_r{}^\vf b_\vf{}^s=&\dl_r^s.
\end{alignat*}
Note that elements of the inverse matrix depend in general on all coordinates
$y$, $z$, and parameters $c$. Matrix $B$ must be nondegenerate in order for
Eq.~(\ref{anfkes}) to have unique solution, and
\begin{equation}                                                  \label{abcvdt}
  g^{\mu\mu}=b_{\Sn+\Sm\,\Sn+\Sm}{}^{\mu\mu}(y,z),\qquad
  2g^{\al\phi}c_\al=b_{\Sn+\Sm\,\Sn+\Sm}{}^\phi(y,z,c),
\end{equation}
because constant $E$ enters the left hand side of Eq.~(\ref{anfkes}) only
through parameter $d_{\Sn+\Sm\,\Sn+\Sm}$. This implies that elements
$b_{\Sn+\Sm\,\Sn+\Sm}{}^{\mu\mu}$ do not depend on $c$, and
$b_{\Sn+\Sm\,\Sn+\Sm}{}^\phi$ are linear in $c$:
\begin{equation*}
  b_{\Sn+\Sm\,\Sn+\Sm}{}^\phi=b_{\Sn+\Sm\,\Sn+\Sm}{}^{\al\phi}c_\al.
\end{equation*}
Substitution of obtained expressions $g^{\mu\nu}$ and $g^{\al\phi}c_\al$ into
the Hamilton--Jacobi equation (\ref{anfkes}) yields
\begin{equation}                                                  \label{erfioi}
  g^{\al\bt}c_\al c_\bt+b_{\Sn+\Sm\,\Sn+\Sm}{}^{\mu\nu}k_{\mu\nu}
  +b_{\Sn+\Sm\,\Sn+\Sm}{}^{\al\phi}c_\al l_{\phi}=0,
\end{equation}
which must be fulfilled for all $c$. Therefore functions $k_{\mu\mu}$ and
$l_\phi$ must be quadratic and linear in $c$, respectively. It is sufficient to
choose them homogeneous
\begin{equation}                                                  \label{abcvdg}
  k_{\mu\mu}=k^{\al\bt}_{\mu\nu}c_\al c_\bt,\qquad l_\vf=l^\al_\vf c_\al,
\end{equation}
where functions $k^{\al\bt}_{\mu\nu}(y^\mu)$ and $l^\al_\vf(z^\vf)$ do not
depend on parameters $c$. Functions $k^{\al\bt}_{\mu\nu}=k^{\bt\al}_{\mu\nu}$
can differ from zero for $\al\ne\bt$. Indices of $k$ and $l$ are written one
over the other because they will never be raised or lowers.

To simplify separable metric, we perform canonical transformation
$(x^\al,z^\vf,p_\al,p_\vf)\mapsto(X^\al,Z^\vf,P_\al,P_\vf)$ with generating
function
\begin{equation}                                                  \label{ebddfr}
  \hphantom{\qquad\qquad\ns(\al)}
  S_2:=x^\al P_\al+z^\vf P_\vf+\int\!\!dz^\vf l^\al_\vf P_\al,
  \quad\quad\ns(\vf).
\end{equation}
Then
\begin{equation}                                                  \label{abdfgr}
\begin{aligned}
  p_\al=&\frac{\pl S_2}{\pl x^\al}=P_\al, & \qquad
  p_\vf=&\frac{\pl S_2}{\pl z^\vf}=P_\vf+l^\al_\vf P_\al,
\\
  X^\al=&\frac{\pl S_2}{\pl p_\al}=x^\al+\int\!\!dz^\vf l^\al_\vf, &
  Z^\vf=&\frac{\pl S_2}{\pl p_\vf}=z^\vf,\quad\quad\ns(\vf).
\end{aligned}
\end{equation}
Afterwards relation (\ref{nksjkh}) in the Hamiltonian formulation becomes
\begin{equation*}
  P_\vf=b_\vf{}^{ij}d_{ij}+b_\vf{}^r a_r,
\end{equation*}
where equalities $p_\al=P_\al=c_\al$ are used. Therefore we can set
$l^\al_\vf\equiv0$ without loss of generality. Now nontrivial blocks of inverse
metric (\ref{abcvdt}) take the form
\begin{equation}                                                  \label{anvbfg}
\begin{split}
  g^{\al\bt}=&-b_{\Sn+\Sm\,\Sn+\Sm}{}^{\mu\nu}k^{\al\bt}_{\mu\nu},
\\
  g^{\mu\mu}=&~~b_{\Sn+\Sm\,\Sn+\Sm}{}^{\mu\mu},
\\
  2g^{\al\phi}=&~~b_{\Sn+\Sm\,\Sn+\Sm}{}^{\al\vf},
\end{split}
\end{equation}
where functions $k^{\al\bt}_{\mu\nu}(y^\mu)$ are arbitrary, and functions
$b_{\Sn+\Sm\,\Sn+\Sm}{}^{\mu\mu}$ and $b_{\Sn+\Sm\,\Sn+\Sm}{}^{\al\vf}$
are defined by matrix (\ref{andkfj}). Note that block $g^{\al\bt}$ may be
degenerate due to the presence of rectangular block $g^{\al\vf}$. In section
\ref{sbbcdh} this is not allowed because metric (\ref{anncbv}) becomes
degenerate.

Moreover, using the canonical transformation of type (\ref{abdvfi}), one of
nonzero elements in each row of matrix $b_{\mu\mu}{}^{ii}$ can be set to unity.

Thus we have solved the functional Hamilton--Jacobi equation.

\begin{theorem}                                                   \label{tiswgk}
Let separable metric be of type $[\Sn,\Sm,n-\Sn-\Sm]_2$. Then there is such set
of parameters (\ref{nsldkj}) and separating functions (\ref{anbcfr}),
(\ref{nksjkh}), that canonical separable metric has block form (\ref{aodewq})
with blocks (\ref{anvbfg}), where functions
$k^{\al\bt}_{\mu\nu}(y^\mu)=k^{\bt\al}_{\mu\nu}(y^\mu)$ are arbitrary, and
functions $b_{\Sn+\Sm\,\Sn+\Sm}{}^{\mu\mu}$ and
$b_{\Sn+\Sm\,\Sn+\Sm}{}^{\al\vf}$ are elements of the $(\Sn+\Sm)$-th row of
the matrix inverse to matrix $B$ (\ref{andkfj}). Matrix $B$ must be
nondegenerate with arbitrary elements, such that elements of each row depend on
single coordinates and parameters $(c_\al)$ in such a way that all elements of
the $(\Sn+\Sm)$-th row of matrix $B^{-1}$ are nonzero, elements
$b_{\Sn+\Sm\,\Sn+\Sm}{}^{\mu\mu}$ do not depend on $(c_\al)$, and elements
$b_{\Sn+\Sm\,\Sn+\Sm}{}^\phi$ are linear in parameters $(c_\al)$. Arbitrary
functions must also produce nondegenerate separable metric and provide real
solutions of Eqs.~(\ref{anbcfr}) for $W'_\mu$ and $W'_\vf$. In addition,
conservation laws in the Hamiltonian formulation have the form
\begin{equation}                                                  \label{asbdgd}
\begin{split}
  p_\al=&c_\al,
\\
  b_{ii}{}^{\mu\nu}\big(p_\mu p_\nu-k^{\al\bt}_{\mu\nu}c_\al c_\bt\big)
  +b_{ii}{}^\vf p_\vf=&d_{ii},
\\
  b_r{}^{\mu\nu}\big(p_\mu p_\nu-k^{\al\bt}_{\mu\nu}c_\al c_\bt\big)
  +b_r{}^\vf p_\vf=&a_r.
\end{split}
\end{equation}
\end{theorem}
In general, there are $\Sn$ linear and $n-\Sn$ quadratic conservation laws for
separable metrics of type $[\Sn,\Sm,n-\Sn-\Sm]_2$ for $\Sm\ge1$.

Let us consider separable metrics of type $[\Sn,\Sm,n-\Sn-\Sm]_1$, i.g.\ energy
$E$ enters the group of parameters for coisotropic coordinates (\ref{ndhtgr}).
Then solution of the Hamilton--Jacobi equation repeats all previous steps. The
only difference is the change of relations (\ref{abcvdt}). Blocks of separable
metric after the canonical transformation (\ref{ebddfr}) become
\begin{equation}                                                  \label{abcvdh}
  g^{\mu\mu}=b_n{}^{\mu\mu}(y,z),\qquad 2g^{\al\phi}=b_n{}^{\al\phi}(y,z).
\end{equation}
Therefore we only formulate the result.

\begin{theorem}                                                   \label{tnbdhg}
Let separable metric be of type $[\Sn,\Sm,n-\Sn-\Sm]_1$. Then there is such set
of parameters (\ref{ndhtgr}) and separating functions (\ref{anbcfr}),
(\ref{nksjkh}), that canonical separable metric has block form (\ref{aodewq})
with blocks
\begin{equation}                                                  \label{ahgdyu}
\begin{split}
  g^{\al\bt}=&-b_n{}^{\mu\nu}k^{\al\bt}_{\mu\nu},
\\
  g^{\mu\mu}=&~~b_n{}^{\mu\mu},
\\
  2g^{\al\vf}=&~~b_n{}^{\al\vf},
\end{split}
\end{equation}
where functions $k^{\al\bt}_{\mu\nu}(y^\mu)=k^{\bt\al}_{\mu\nu}(y^\mu)$ are
arbitrary, and functions $b_{nn}{}^{\mu\mu}$ and $b_{nn}{}^{\al\vf}$  are
elements of the $n$-th row of the matrix inverse to matrix $B$ (\ref{andkfj}).
Matrix $B$ must be nondegenerate with arbitrary elements such that elements of
each row depend on single coordinates and parameters $(c_\al)$ in such a way
that all elements of the last row of matrix $B^{-1}$ are nonzero, elements
$b_{nn}{}^{\mu\mu}$ do not depend on $(c_\al)$, and $b_{nn}{}^\phi$ are linear
in parameters $(c_\al)$. Arbitrary functions must also produce nondegenerate
separable metric and provide real solutions of Eqs.~(\ref{anbcfr}) for $W'_\mu$
and $W'_\vf$. In addition, conservation laws in the Hamiltonian formulation have
form (\ref{asbdgd}).
\end{theorem}

Thus we looked over all possible classes of separable metrics, introduced sets
of independent parameters and separating functions and explicitly separated all
variables. The used method is constructive and allows to write down all $n$
conservation laws (\ref{asbdgd}) for geodesics. We showed that there exists such
choice of independent parameters and separating functions in complete integrals
that conservation laws are at most quadratic in momenta: (i), there are linear
conservation laws corresponding to commuting Killing vector fields; (ii), there
are indecomposable quadratic conservation laws; (iii), there may exist linear
conservation laws for coisotropic coordinates which are not related to Killing
vectors. The last possibility arises only for indefinite metrics which may have
zeroes on the diagonal.

We see that constructed conservation laws are at most quadratic.
\begin{theorem}
Let the Hamilton--Jacobi equation for geodesics (\ref{ehdggt}) admit complete
separation of variables. Then there exists such choice of independent parameters
and separating functions in complete integrals that corresponding conservation
laws are at most quadratic in momenta.
\end{theorem}
This statement is important because for another choice of parameters and
separating functions we cannot guarantee the absence of higher order
conservation laws.

At the end we prove that any additional conservation law is a function of
conservation laws described in theorems (\ref{tiswgk}) and (\ref{tnbdhg}).
To this end we square equality $p_\al=c_\al$. Then all conservation laws become
quadratic (remember that functions $b_{ii}{}^\vf$ and $b_r{}^\vf$ are linear in
$c_\al$ and, consequently, in momenta $p_\al$). Part of these conservation laws
may be decomposable but this is not essential. We simplify notation:
\begin{equation*}
\begin{split}
  (x^\al,y^\mu,z^\vf)\quad\mapsto&\quad(q^\al),
\\
  (p_\al,p_\mu,p_\vf)\quad\mapsto&\quad(p_\al),
\\
  (c_\al^2,d_{ij},a_r)\quad\mapsto&\quad(c_\Sa),
\end{split}
\end{equation*}
where indices in the right hand sides take all values from $1$ to $n$:
$\al=1,\dotsc,n$ and $\Sa=1,\dotsc,n$. The conservation laws take the form
\begin{equation*}
  c_\Sa=F_\Sa(q,p)=\frac12 f_\Sa^{\al\bt}p_\al p_\bt,
\end{equation*}
where $f_\Sa^{\al\bt}(q)$ are some functions only on coordinates $q$. They
contain the Hamiltonian among themselves: for separable metrics of type
$1$ and $2$, it is $F_n$ and $F_{\Sn+\Sm}$, respectively. Perform the canonical
transformation $(q^\al,p_\al)\mapsto(Q^\Sa,P_\Sa)$ with generating functions
$S_2:=W(q,P)$ to action-angle variables. Then
\begin{equation*}
  F_\Sa=P_\Sa,\qquad\forall\Sa=1,\dotsc n,
\end{equation*}
in new variables. Assume that there is additional involutive conservation law
\begin{equation*}
  G:=\sum_{\Sa=1}^n G^\Sa(Q) P_\Sa^m=\const,\qquad m\in\MR,
\end{equation*}
with some differentiable functions $G^\Sa$ only from new coordinates, and where
$m$ is some real number. In particular, $G$ may be homogeneous polynomial of any
order $m$ on new momenta. Involutivity means that
\begin{equation*}
  [G,F_\Sa]=[G,P_\Sa]=\sum_{\Sb=1}^n\frac{\pl G^\Sb}{\pl Q^\Sa}P_\Sb^m=0.
\end{equation*}
This equality must hold for all momenta and values of index $\Sa$. Therefore
functions $G^\Sa$ do not depend on coordinates and, consequently, there is
functional dependence between integrals $(F_\Sa,G)$ because the number of
functions $F_\Sa$ in maximal and equals to the number of momenta, i.e.\ any
additional conservation law is some function on integrals $(F_\Sa)$ at least
locally, $G=G(F)$. Thus we proved
\begin{theorem}                                                   \label{tssuyy}
Let variables in the Hamilton--Jacobi equation for geodesics be completely
separable, and all conservation laws $F_\Sa$ are found according to
theorems (\ref{tiswgk}) and (\ref{tnbdhg}). Then any additional involutive
conservation law $G$ for geodesics is some function $G=G(F)$ at least locally.
\end{theorem}

This solves completely the St\"ackel problem for metrics of arbitrary signature
on manifolds of any dimension. All separable metrics are divided into
equivalence classes. Metrics in each class are related by canonical
transformations and nondegenerate transformations of parameters which do not
involve coordinates. There is the canonical (the simplest) separable metric in
each equivalence class. Its form is given by theorems \ref{tiswgk} and
\ref{tnbdhg}. Corresponding conservation laws are at most quadratic in momenta.
For other choices of parameters conservation laws may have higher order but they
are functionally dependent on conservation laws listed in theorems \ref{tiswgk}
and \ref{tnbdhg} due to theorem \ref{tssuyy}. Matrix $B$ in general theorems is
constructed as matrix (\ref{isndgf}). The proved theorems are constructive. In
next sections we list all separable metrics in two, three, and four dimensions
as examples.
\section{Separation of variables in two dimensions}
Two dimensional manifolds provide the simplest examples of separable metrics for
the Hamilton-- Jacobi equation for geodesics which have important features
present in higher dimensions. Separation of variables in two dimensions is
considered in example \ref{ecddsa} in detail. Therefore we only formulate
results.

There are only three classes\- of different separable metrics.

1) {\bf Class $[2,0,0]$.} Two commuting Killing vectors. Coordinates and
parameters:
\begin{equation*}
  (x^\al,y^\mu,z^\vf)\mapsto(x^1,x^2),\qquad
  (c_\al,d_{ij},a_r)\mapsto(c_1,c_2).
\end{equation*}
This case is described by theorem \ref{tdjkfh}. In canonical form, the inverse
separable metric is (pseudo)Euclidean:
\begin{equation}                                                  \label{abcvdf}
  g^{\al\bt}=\eta^{\al\bt},
\end{equation}
where $\eta^{\al\bt}$ is either Euclidean or Lorentzian metric. The
Hamilton--Jacobi equation has the form
\begin{equation*}
  \eta^{\al\bt}W'_\al W'_\bt=\eta^{\al\bt}c_\al c_\bt.
\end{equation*}
Variables are separated as
\begin{equation*}
  W'_\al=c_\al.
\end{equation*}
The system has two linear conservation laws:
\begin{equation}                                                  \label{issnsf}
  p_\al=c_\al,
\end{equation}
and respective coordinates $x^1,x^2$ are cyclic.
Separable metric (\ref{abcvdf}) has two independent commuting Killing vector
fields $\pl_1$ and $\pl_2$.

2) {\bf Class $[1,1,0]_2$.} One Killing vector and one indecomposable quadratic
conservation law. Coordinates and parameters:
\begin{equation*}
  (x^\al,y^\mu,z^\vf)\mapsto(x^1:=x,y^2:=y),\qquad
  (c_\al,d_{ij},a_r)\mapsto(c_1:=c,d_{22}:=2E).
\end{equation*}
This case in more generality is considered in example \ref{esjjhg}. Canonical
separable (pseudo)\-Riemannian metric has form (\ref{aghsjy})
\begin{equation}                                                  \label{anbcgf}
  g^{**}=\begin{pmatrix} -k(y) & 0 \\ 0 & 1 \end{pmatrix},
\end{equation}
where $k(y)\ne0$ is an arbitrary function. For $k<0$ and $k>0$, we have
Riemannian and Lorentzian metrics, respectively. The Hamilton--Jacobi equation
is
\begin{equation*}
  -k W^{\prime2}_1+W^{\prime2}_2=2E.
\end{equation*}
Variables are separated in the following way
\begin{equation*}
  W_1=c,\qquad W^{\prime2}_2=2E+kc^2>0.
\end{equation*}
In the Hamiltonian formulation, we have one linear and one quadratic
conservation laws:
\begin{equation*}
  p_1=c,\qquad p_2^2-k(y)c^2=2E.
\end{equation*}
Canonical separable metric (\ref{anbcgf}) is parameterized by one arbitrary
function $k(y)\ne0$, satisfying inequality $2E+kc^2>0$ for fixed $E$ and $c$.

3) {\bf Class $[0,2,0]_2$.} Tow quadratic conservation laws. Coordinates and
parameters:
\begin{equation*}
  (x^\al,y^\mu,z^\vf)\mapsto(y^1,y^2),\qquad
  (c_\al,d_{ij},a_r)\mapsto(d_{11}:=d,d_{22}:=2E).
\end{equation*}
This case is described by theorem \ref{tgdjke}. In canonical form, the separable
metric is given by Eq.~(\ref{eqkduy}).
Variables are completely separated (\ref{eodgfu}), where matrix $b$ is given
by Eq.~(\ref{andbcf}). Conservation laws are quadratic (\ref{amnssg}).
The canonical separable metric is parameterized by four arbitrary functions
$\phi_{11}(y^1)$, $\phi_{12}(y^1)$, $\phi_{21}(y^2)$, and $\phi_{22}(y^2)$,
whose restrictions are described in example \ref{ecddsa} in detail. This class
of metrics contains two dimensional Liouville systems (\ref{ewtref}).

4) {\bf Class $[1,0,1]_1$.} One Killing vector and one coisotropic coordinate.
Coordinates and parameters:
\begin{equation*}
  (x^\al,y^\mu,z^\vf)\mapsto(x^1:=x,z^2:=z),\qquad
  (c_\al,d_{ij},a_r)\mapsto(c_1:=c,a_2:=2E).
\end{equation*}
Appearance of zeroes on the diagonal is possible only for Lorentzian signature
metrics and one Killing vector. This case is described by theorem \ref{tdhfgd}.
Canonical separable metric has form (\ref{invbff}) where off diagonal block
$b_n{}^{\al\phi}$ is reduced to one nonzero function $b(z)$, which can be
transformed to unity by the canonical transformation, i.e.\
\begin{equation}                                                  \label{abcvdl}
  g^{**}=\begin{pmatrix} 0 & 1 \\ 1 & 0 \end{pmatrix}.
\end{equation}
This metric produces two linear conservation laws
\begin{equation}                                                  \label{awtruy}
  p_1=c,\qquad cp_2=E
\end{equation}
We see that separation of variables in two dimensions corresponds to two
commuting Killing vectors $\pl_x$ and $\pl_z$, i.e.\ coordinates $x$, $z$ are
null. Therefore separable metric (\ref{abcvdl}) is equivalent to the Lorentz
metric, and respective classes are also equivalent $[1,0,1]_1\sim[2,0,0]$.

All three classes of separable metrics were known already to St\"ackel
\cite{Stacke91}. However the case 4) with coisotropic coordinate was not
described. Moreover, we proved that there no other separable metrics.
\section{Separation of variables in three dimensions}             \label{sdsjuj}
There are six different classes of separable metrics.

1) {\bf Class $[3,0,0]$.} Three commuting Killing vectirs. Coordinates and
parameters:
\begin{equation*}
  (x^\al,y^\mu,z^\vf)\mapsto(x^1,x^2,x^3),\qquad
  (c_\al,d_{ij},a_r)\mapsto(c_1,c_2,c_3).
\end{equation*}
In canonical form, the separable metric is (pseudo)Euclidean. Separation of
variables and conservation laws have the same form as in two dimensions in case
$[2,0,0]$, but now indices $\al=1,2,3$ take three values, and all Cartesian
coordinates $x^\al$ are cyclic.

2) {\bf Class $[2,1,0]_2$.} Two commuting Killing vectors and one quadratic
conservation law. Coordinates and parameters:
\begin{equation*}
  (x^\al,y^\mu,z^\vf)\mapsto(x^1,x^2,y^3:=y),\qquad
  (c_\al,d_{ij},a_r)\mapsto(c_1,c_2,d_{33}:=2E).
\end{equation*}
Canonical separable metric has the form (\ref{aghsjy}):
\begin{equation}                                                  \label{abcvdi}
  g^{**}=\begin{pmatrix} g^{\al\bt}(y) & 0 \\ 0 & 1\end{pmatrix},
  \qquad \al,\bt=1,2,
\end{equation}
where $g^{\al\bt}$ is an arbitrary symmetric nondegenerate matrix. The
Hamilton--Jacobi equation and separation of variables are as follows
\begin{equation*}
\begin{split}
  &g^{\al\bt}W'_\al W'_\bt+W^{\prime2}_3=2E,
\\
  &W'_\al=c_\al,\qquad\qquad W^{\prime2}_3=2E-g^{\al\bt}c_\al c_\bt.
\end{split}
\end{equation*}
Conservation laws are
\begin{equation}                                                  \label{edjhyt}
  p_\al=c_\al,\qquad g^{\al\bt}(y)p_\al p_\bt+p_3^2=2E.
\end{equation}
In this case, the canonical separable metric is parameterized by three
arbitrary functions of single argument in matrix $g^{\al\bt}$ with restriction
$\det g^{\al\bt}\ne0$. Moreover, condition $W^{\prime2}_3\ge0$ also restricts
arbitrary functions for fixed $E$ and $c$. Depending on matrix $g^{\al\bt}$
separable metric (\ref{abcvdi}) may have arbitrary signature. In particular,
when matrix $g^{\al\bt}$ is constant, we return to class $[3,0,0]$.

3) {\bf Class $[1,2,0]_2$}. One Killing vector and two quadratic conservation
laws. Coordinates and parameters:
\begin{equation*}
  (x^\al,y^\mu,z^\vf)\mapsto(x^1,y^2,y^3),\qquad
  (c_\al,d_{ij},a_r)\mapsto(c_1:=c,d_{22}:=d,d_{33}:=2E).
\end{equation*}
This case is described by theorem \ref{theojk}. It is new, and we consider it
in more detail. First, we define matrix $b$:
\begin{equation*}
  b_{\mu\mu}{}^{ii}=\begin{pmatrix} \phi_{22}(y^2) & \phi_{23}(y^2) \\
  \phi_{32}(y^3) & \phi_{33}(y^3) \end{pmatrix}, \qquad
  b_{ii}{}^{\mu\mu}=\frac1{\det b}\begin{pmatrix} \phi_{33} & -\phi_{23} \\
  -\phi_{32} & \phi_{22} \end{pmatrix},
\end{equation*}
where $\det b=\phi_{22}\phi_{33}-\phi_{23}\phi_{32}\ne0$. After canonical
transformation of type (\ref{abdvfi}), One element in each row of matrix $b$
can be transformed to $\pm1$. Let
\begin{equation*}
  \phi_{22}=1,\qquad \phi_{32}=-1.
\end{equation*}
Then
\begin{equation}                                                  \label{ebdpoi}
  b_{\mu\mu}{}^{ii}=\begin{pmatrix} ~~1 & \phi_{23} \\
  -1 & \phi_{33} \end{pmatrix}, \qquad
  b_{ii}{}^{\mu\mu}=\frac1{\det b}\begin{pmatrix} \phi_{33} & -\phi_{23} \\
  1 & 1 \end{pmatrix},
\end{equation}
where $\det b=\phi_{23}+\phi_{33}$. Diagonal metric elements corresponding to
quadratic conservation laws are
\begin{equation}                                                  \label{ewshdg}
  g^{22}=b_{33}{}^{22}=\frac1{\phi_{23}+\phi_{33}},\qquad
  g^{33}=b_{33}{}^{33}=\frac1{\phi_{23}+\phi_{33}}.
\end{equation}
Element $g^{11}$ of inverse metric has form (\ref{anncbv})
\begin{equation*}
  g^{11}=-b_{33}^{\mu\nu}k^{11}_{\mu\nu}=-\frac1{\phi_{23}+\phi_{33}}
  (k^{11}_{22}+k^{11}_{33}).
\end{equation*}
Simplify notation
\begin{equation*}
  \phi_{23}(y^2):=\phi_2(y^2),\quad \phi_{33}(y^3):=\phi_3(y^3),\qquad
  k^{11}_{22}(y^2):=-k_2(y^2)\quad k^{11}_{33}:=-k_3(y^3).
\end{equation*}
Now the canonical separable metric takes the form
\begin{equation}                                                  \label{aswgre}
  g^{**}=\frac1{\phi_2+\phi_3}\begin{pmatrix} k_2+k_3 & 0 & 0 \\
  0 & 1 & 0 \\ 0 & 0 & 1 \end{pmatrix}.
\end{equation}
Variables in the Hamilton--Jacobi equation
\begin{equation*}
  \frac1{\phi_2+\phi_3}\big[(k_2+k_3)W^{\prime2}_1+W^{\prime2}_2+W^{\prime2}_3
  \big]=2E
\end{equation*}
are completely separated:
\begin{equation}                                                  \label{aweeth}
\begin{split}
  W'_1=&~~c,
\\
  W^{\prime2}_2=&~~d+2\phi_2E-k_2 c^2,
\\
  W^{\prime2}_3=&-d+2\phi_3 E-k_3 c^2.
\end{split}
\end{equation}
These relations yield conservation laws
\begin{equation}                                                  \label{aljghg}
\begin{split}
  p_1=&c,
\\
  \frac1{\phi_2+\phi_3}\left[\phi_3 p_2^2-\phi_{2}p_3^2
  +\big(\phi_3 k_2-\phi_2 k_3\big)p_1^2\right]=&d,
\\
  \frac1{\phi_2+\phi_3}\left[p_2^2+p_3^2+\big(k_2+k_3\big)p_1^2\right]=&2E.
\end{split}
\end{equation}
Thus canonical separable metric (\ref{aswgre}) is parameterized by four
functions of single argument $\phi_{2,3}$ and $k_{2,3}$. They have to produce
nondegenerate metric, and Eq.~(\ref{aweeth}) must admit real solutions for
separating functions $W'_\mu$. In general, there are two indecomposable
quadratic and one linear conservation laws.

4) {\bf Class $[1,1,1]_2$.} One Killing vector, one indecomposable quadratic
conservation law, and one coisotropic coordinate. Coordinates and parameters:
\begin{equation*}
  (x^\al,y^\mu,z^\vf)\mapsto(x^1:=x,y^2:=y,z^3:=z),\qquad
  (c_\al,d_{ij},a_r)\mapsto(c_1:=c,d_{22}:=2E,a_3:=a).
\end{equation*}
This case is described by theorem \ref{tiswgk}. As always, we start with
arbitrary $(2\times2)$ matrix (\ref{andkfj}). To simplify notation, let
\begin{equation*}
  b_{22}{}^{22}\equiv1,\qquad b_{22}{}^3\equiv\phi_2(y),\qquad
  b_3{}^{22}\equiv\phi_3/c,\qquad b_3{}^3\equiv1/c.
\end{equation*}
The first and the last equality can be always achieved by suitable canonical
transformation. The dependence on $c$ follows from the independence of metric
components on parameters. Then
\begin{equation}                                                  \label{iddkll}
  B=\begin{pmatrix} 1 & \phi_2 \\ \phi_3/c & 1/c
  \end{pmatrix}\qquad\Rightarrow\qquad
  B^{-1}=\frac1{1-\phi_2\phi_3}\begin{pmatrix} 1 & -c\phi_2 \\
  -\phi_3 & c \end{pmatrix},
\end{equation}
where $\phi_2(y)$ and $\phi_3(z)$ are arbitrary functions on single coordinates.
Equations (\ref{anvbfg}) imply expression for the canonical separable metric
\begin{equation}                                                  \label{abcvdg}
  g^{**}=\frac1{1-\phi_2\phi_3}
  \begin{pmatrix}-k_2 & 0 & -\phi_2/2 \\ 0 & 1 & 0 \\
  -\phi_2/2 & 0 & 0 \end{pmatrix},
\end{equation}
where $k_2(y)$ is an arbitrary function. The Hamilton--Jacobi equation after
substitution $W'_1=c$ takes the form
\begin{equation*}
  \frac1{1-\phi_2\phi_3}\left(-k_2c^2+W^{\prime2}_2-\phi_2cW'_3\right)=2E.
\end{equation*}
Variables are separated in the following form
\begin{equation}                                                  \label{avbsfd}
\begin{split}
  W^{\prime2}_2=&2E+\phi_2 a+k_2c^2,
\\
  W'_3=&\frac1c(2\phi_3E+a).
\end{split}
\end{equation}
Conservation laws (\ref{asbdgd}) after substitution $p_1=c$ are
\begin{equation}                                                  \label{anbghy}
\begin{split}
  \frac1{1-\phi_2\phi_3}\big(p_2^2-k_2c^2-\phi_2cp_3\big)=&2E,
\\
  \frac 1{1-\phi_2\phi_3}\big[-\phi_3(p_2^2-k_2c^2)+cp_3\big]=&a.
\end{split}
\end{equation}

Canonical separable metric (\ref{abcvdg}) is parameterized by three arbitrary
functions $\phi_2(y)$, $\phi_3(z)$, and $k_2(y)$. Determinant of the inverse
metric (\ref{abcvdg}) is
\begin{equation*}
  \det g^{**}=-\frac{\phi_2^2}{4(1-\phi_2\phi_3)^3}\ne0,\infty.
\end{equation*}
Arbitrary functions must be restricted $\phi_2\ne0$ and $\phi_2\phi_3\ne1$,
because otherwise the determinant is degenerate. Moreover functions $\phi_2$
and $k_2$ are to be chosen in such a way as to provide real solution for the
first equation (\ref{avbsfd}) with respect to $W'_2$. Depending on arbitrary
functions the separable metric may any signature.

5) {\bf Class $[1,1,1]_1$.} One Killing vector, one indecomposable quadratic
conservation law, and one coisotropic coordinate. Coordinates and parameters:
\begin{equation*}
  (x^\al,y^\mu,z^\vf)\mapsto(x^1:=x,y^2:=y,z^3:=z),\qquad
  (c_\al,d_{ij},a_r)\mapsto(c_1:=c,d_{22}:=d,a_3:=2E).
\end{equation*}
This case is given by theorem \ref{tnbdhg}. Matrix $B$ has the same form as in
case $[1,1,1]_2$ (\ref{iddkll}), however the canonical separable metric is
defined by the last row of matrix $B^{-1}$:
\begin{equation}                                                  \label{alcvdg}
  g^{**}=\frac1{1-\phi_2\phi_3}
  \begin{pmatrix} \phi_3k_2 & 0 & 1/2 \\ 0 & -\phi_3 & 0 \\
  1/2 & 0 &0\end{pmatrix}.
\end{equation}
After substitution $W'_1\equiv c$, the Hamilton--Jacobi equation take the form
\begin{equation*}
  \frac1{1-\phi_2\phi_3}\big(\phi_3k_2c^2-\phi_3W^{\prime2}_2+cW'_3\big)=2E.
\end{equation*}
Variables are separated as
\begin{equation}                                                  \label{avbsfl}
\begin{split}
  W^{\prime2}_2=&d+2\phi_2E+k_2c^2,
\\
  W'_3=&\frac1c({\phi_3}d+2E).
\end{split}
\end{equation}
Conservation laws (\ref{asbdgd}) take the form
\begin{equation}                                                  \label{avbdgt}
\begin{split}
  \frac1{1-\phi_2\phi_3}\big[p_2^2-k_2c^2-\phi_2cp_3\big]=&d,
\\
  \frac1{1-\phi_2\phi_3}\big[-\phi_3(p_2^2-k_2c^2)+cp_3\big]=&2E.
\end{split}
\end{equation}

So, canonical separable metric (\ref{alcvdg}) is parameterized by three
arbitrary functions $\phi_2(y)$, $\phi_3(z)$, and $k_2(y)$. They are restricted
by nondegeneracy of the determinant
\begin{equation*}
  \det g^{**}=\frac{\phi_3}{4(1-\phi_2\phi_3)^3}\ne0,\infty
\end{equation*}
which imply $\phi_3\ne0$ and $\phi_2\phi_3\ne1$. Moreover, functions $\phi_2$
and $k_2$ have to provide existence of real solutions for $W'_2$ in the first
equation (\ref{avbsfl}). Depending on arbitrary functions, the separable metric
may have any signature.

6) {\bf Class $[0,3,0]_2$.} Absence of Killing vectors and coisotropic
coordinates. Coordinates and parameters:
\begin{equation*}
  (x^\al,y^\mu,z^\vf)\mapsto(y^1,y^2,y^3),\qquad
  (c_\al,d_{ij},a_r)\mapsto(d_{11}:=d_1,d_{22}:=d_2,d_{33}:=2E).
\end{equation*}
This case is described by theorem \ref{tgdjke}. After the canonical
transformation with generating function (\ref{abdvfi}) matrix $b$ and its
inverse are
\begin{equation}                                                  \label{idlkih}
  b_{\mu\mu}{}^{ii}=
  \begin{pmatrix}1 & b_{12}(y^1) & b_{13}(y^1)\\ b_{21}(y^2) & 1 & b_{23}(y^2)\\
  b_{31}(y^3) & b_{32}(y^3) & 1\end{pmatrix},\qquad
  b_{ii}{}^{\mu\mu}=\frac1\vartriangle
  \begin{pmatrix} \vartriangle_{11} & \vartriangle_{21} & \vartriangle_{31} \\
  \vartriangle_{12} & \vartriangle_{22} & \vartriangle_{32} \\
  \vartriangle_{13} & \vartriangle_{23} & \vartriangle_{33} \end{pmatrix},
\end{equation}
where $\vartriangle:=\det b_{\mu\mu}{}^{ii}$ and symbols $\vartriangle_{\mu i}$
denote cofactors of elements $b_{\mu\mu}{}^{ii}$. This matrix produces the
diagonal metric (\ref{adfghf})
\begin{equation}                                                  \label{avcfsg}
  g^{**}=\frac1\vartriangle \begin{pmatrix} \vartriangle_{13} & 0 & 0 \\
  0 & \vartriangle_{23} & 0 \\ 0 & 0 & \vartriangle_{33} \end{pmatrix}.
\end{equation}
The respective Hamilton--Jacobe equation becomes
\begin{equation*}
  \frac1\vartriangle\big(\vartriangle_{13}W^{\prime2}_1
  +\vartriangle_{23}W^{\prime2}_2+\vartriangle_{33}W^{\prime2}_3\big)=2E.
\end{equation*}
Variables are completely separated in the way
\begin{equation}                                                  \label{abbcvf}
\begin{split}
  W^{\prime2}_1=&d_1+b_{12}d_2+2b_{13}E,
\\
  W^{\prime2}_2=&b_{21}d_1+d_2+2b_{23}E,
\\
  W^{\prime2}_3=&b_{31}d_1+b_{32}d_2+2E.
\\
\end{split}
\end{equation}
All three conservation laws are quadratic in general
\begin{equation}                                                  \label{abcvdo}
\begin{split}
  \frac1\vartriangle\big(\vartriangle_{11}p^2_1+\vartriangle_{21}p^2_2
  +\vartriangle_{31}p^2_3\big)=&d_1,
\\
  \frac1\vartriangle\big(\vartriangle_{12}p^2_1+\vartriangle_{22}p^2_2
  +\vartriangle_{32}p^2_3\big)=&d_2,
\\
  \frac1\vartriangle\big(\vartriangle_{13}p^2_1+\vartriangle_{23}p^2_2
  +\vartriangle_{33}p^2_3\big)=&2E.
\end{split}
\end{equation}

This case includes the Liouville system (see example \ref{ejsdfw}). Indeed,
take matrix $b$ in the form
\begin{equation*}
  b_{\mu\mu}{}^{ii}:=\begin{pmatrix} 1 & 0 &\phi_1(y^1) \\ 0 & 1 & \phi_2(y^2)
  \\-1 & -1 & \phi_3(y^3)\end{pmatrix}\quad\Rightarrow\quad
  b_{ii}{}^{\mu\mu}=\frac1{\phi_1+\phi_2+\phi_3}
  \begin{pmatrix} \phi_2+\phi_3 & -\phi_1 & -\phi_1 \\
  -\phi_2 & \phi_1+\phi_3 & -\phi_2 \\ 1 & 1 & 1 \end{pmatrix}.
\end{equation*}
Then conservation laws are
\begin{equation}                                                  \label{anskiu}
\begin{split}
  p^2_1-2\phi_1 E=&d_1,
\\
  p^2_2-2\phi_2 E=&d_2,
\\
  \frac1{\phi_1+\phi_2+\phi_3}\big(p^2_1+p^2_2+p^2_3\big)=&2E,
\end{split}
\end{equation}
which coincides with Eqs.~(\ref{eddgtr}) and (\ref{essghk}) for $\Theta\equiv0$
up to notation.

Five of six separable metrics in three dimensions were listed in
\cite{KalMil79}, where different technique is used. Types of metrics coincide
with respect to the number of Killing vectors, quadratic conservation laws and
coisotropic coordinates. However computations of the present section allowed us
to find explicitly separable functions $W'_\al$, $\al=1,\dotsc,3$. In addition,
we give more detailed classification indicated by indices $1,2$, and separable
metric of class $[1,1,1]_1$ is not mentioned in paper \cite{KalMil79}.
\section{Separation of variables in four dimensions}
Separable metrics in four dimensions are of great importance in gravity models,
in particular, in general relativity. These metrics has Lorentzian signature,
and coisotropic coordinates may appear. There are ten classes of different
separable metrics in four dimensions.

1) {\bf Class $[4,0,0]$.} Four Killing vectors. Variables are separated in the
same way as in lower dimensions. Canonical separable metric and conservation
laws have form (\ref{abcvdf}) and (\ref{issnsf}), respectively, but now indices
run over all four values, $\al,\bt=1,2,3,4$.

2) {\bf Class $[3,1,0]_2$.} Three commuting Killing vectors and one
indecomposable quadratic conservation law. As in three dimensions, the canonical
separable metric and conservation laws have form (\ref{abcvdi}) and
(\ref{edjhyt}), but indices $\al,\bt=1,2,3$ take more values. The canonical
separable metric is parameterized by six arbitrary functions.

The Kasner solution \cite{Kasner21C} lies in this class.

3) {\bf Class $[2,2,0]_2$.} Two commuting Killing vector fields and two
indecomposable quadratic conservation laws without coisotropic coordinates.
Coordinates and parameters:
\begin{equation*}
  (x^\al,y^\mu,z^\vf)\mapsto(x^1,x^2,y^3,y^4),\qquad
  (c_\al,d_{ij},a_r)\mapsto(c_1,c_2,d_{33}:=d,d_{44}:=2E).
\end{equation*}
This case is described by theorem \ref{theojk}. Matrix $b$ and diagonal metric
elements $g^{33}$ and $g^{44}$ have form (\ref{ebdpoi}) and (\ref{ewshdg}) with
replacement $(2,3)\mapsto(3,4)$. The block of inverse metric $g^{\al\bt}$,
$\al,\bt=1,2$, is
\begin{equation*}
  g^{\al\bt}=-k^{\al\bt}_{\mu\nu}g^{\mu\nu}=\frac1{\phi_3+\phi_4}
  \big(k^{\al\bt}_3+k^{\al\bt}_4\big),\qquad\phi_3+\phi_4\ne0,
\end{equation*}
where $\phi_3(y^3)$, $\phi_4(y^4)$, and $k^{\al\bt}_3(y^3)$, $k^{\al\bt}_4(y^4)$
-- are arbitrary functions of single coordinates. Thus canonical separable
metric has the form
\begin{equation}                                                  \label{ancbfg}
  g^{**}=\frac1{\phi_3+\phi_4}\begin{pmatrix}
  k^{\al\bt}_3+k^{\al\bt}_4 & 0 & 0 \\ 0 & 1 & 0 \\ 0 & 0 & 1 \end{pmatrix}.
\end{equation}
In addition, we have to require fulfillment of two conditions
\begin{equation}                                                  \label{ekdfre}
  \det g^{**}=\frac{(k^{11}_3+k^{11}_4)(k^{22}_3+k^{22}_4)-
  (k^{12}_3+k^{12}_4)(k^{21}_3+k^{21}_4)}{(\phi_3+\phi_4)^4}\ne0,\infty.
\end{equation}
After separation of two coordinates $W'_\al\equiv c_\al$, the Hamilton--Jacobi
equation becomes
\begin{equation}                                                  \label{anbvhf}
  \frac1{\phi_3+\phi_4}\big(k^{\al\bt}_3c_\al c_\bt
  +k^{\al\bt}_4c_\al c_\bt+W^{\prime2}_3+W^{\prime2}_4\big)=2E.
\end{equation}
Variables are completely separated as
\begin{equation}                                                  \label{absgfd}
\begin{split}
  W^{\prime2}_3=&~~d+2E\phi_3-k^{\al\bt}_3c_\al c_\bt,
\\
  W^{\prime2}_4=&-d+2E\phi_4-k^{\al\bt}_4c_\al c_\bt.
\end{split}
\end{equation}
Two conservation laws are quadratic in general
\begin{equation}                                                  \label{abcvdu}
\begin{split}
  \frac1{\phi_3+\phi_4}\big[\phi_4(p_3^2+k_3^{\al\bt}c_\al c_\bt)
  -\phi_3(p_4^2+k_4^{\al\bt}c_\al c_\bt\big]=&d,
\\
  \frac1{\phi_3+\phi_4}\big[p_3^2+k_3^{\al\bt}c_\al c_\bt+p_4^2+k^{\al\bt}_4
  c_\al c_\bt\big]=&2E.
\end{split}
\end{equation}
In this case, the canonical separable metric (\ref{ancbfg}) is parameterized by
eight arbitrary functions, $\phi_3(y^3)$, $\phi_4(y^4)$, $k^{\al\bt}_3(y^3)$,
and $k^{\al\bt}_4(y^4)$ satisfying inequalities (\ref{ekdfre}). Besides, they
must be chosen in such a way that equations for $W'$ (\ref{absgfd}) have real
solutions.

This class of separable metrics includes the Schwarzschild,
Reissner--Nordstr\"om, Kerr, and other famous solutions in general relativity
\cite{Carter68}.

4) {\bf Class $[2,0,2]_1$.} Two commuting Killing vectors and two coisotropic
coordinates. Coordinates and parameters:
\begin{equation*}
  (x^\al,z^\vf)=(x^1,x^2,z^3,z^4),\qquad
  (c_\al,d_{ii},a_r)\mapsto(c_1,c_2,a_3:=a,a_4:=2E).
\end{equation*}
This case is described by theorem \ref{tdhfgd}. It is new and considered in more
detail. Matrix $\phi$ (\ref{irkrjg}) has the block form
\begin{equation}                                                  \label{anbgfd}
  (\phi_\vf{}^r)=\begin{pmatrix} 1 & \phi_3 \\ \phi_4 & 1 \end{pmatrix},
  \qquad\Leftrightarrow\qquad
  (\phi_\vf{}^r)^{-1}=\frac1{1-\phi_3\phi_4}\begin{pmatrix} 1 & -\phi_3 \\
  -\phi_4 & 1 \end{pmatrix},
\end{equation}
where $\phi_3(z^3)$ and $\phi_4(z^4)$ are arbitrary functions of single
variables. Functions $h_\vf$ are
\begin{equation*}
  h_3(z^3,c)=h_3^\al(z^3) c_\al,\qquad h_4(z^4,c)=h_4^\al(z^4)c_\al,\qquad
  \al,\bt=1,2,
\end{equation*}
where $h_3^\al(z^3)$ and $h_4^\al(z^4)$ are linear. Then matrix $b$ has the form
\begin{equation}                                                  \label{aihwkw}
  (b_\vf{}^r)=\begin{pmatrix} \displaystyle\frac1{h_3^\al c_\al} &
  \displaystyle\frac{\phi_3}{h_3^\al c_\al} \\[10pt] \displaystyle
  \frac{\phi_4}{h_4^\al c_\al} & \displaystyle\frac1{h_4^\al c_\al}
  \end{pmatrix}\qquad\Leftrightarrow\qquad
  (b_r{}^\vf)=\frac1{1-\phi_3\phi_4} \begin{pmatrix} h_3^\al c_\al &
  -\phi_3h_4^\al c_\al \\ -\phi_4h_3^\al c_\al & h_4^\al c_\al \end{pmatrix}.
\end{equation}
It implies separable metric
\begin{equation}                                                  \label{aismnd}
  g^{**}=\frac1{1-\phi_3\phi_4}\begin{pmatrix} 0 & 0 & -\phi_4h_3^1 & h_4^1 \\
  0 & 0 & -\phi_4h_3^2 & h_4^2 \\ -\phi_4h_3^1 & -\phi_4h_3^2 & 0 & 0 \\
  h_4^1 & h_4^2 & 0 & 0 \end{pmatrix}.
\end{equation}
The Hamilton--Jacobi equation after separation of the first two coordinates,
\begin{equation}                                                  \label{ikevsf}
  \frac1{1-\phi_3\phi_4}\big(-\phi_3h_3^\al c_\al W'_3+h_4^\al c_\al W'_4\big)
  =2E,
\end{equation}
is completely separated as
\begin{equation}                                                  \label{adwsdf}
\begin{split}
  W'_3=&\frac1{h_3^\al c_\al}(a+2\phi_3 E),
\\
  W'_4=&\frac1{h_4^\al c_\al}(\phi_4 a+2E).
\end{split}
\end{equation}
Conservation laws (\ref{abndjh}) are linear
\begin{equation}                                                  \label{abcnda}
\begin{split}
  \frac1{1-\phi_3\phi_4}\big(h_3^\al c_\al p_3-\phi_3h_4^\al c_\al p_4\big)=&a,
\\
  \frac1{1-\phi_3\phi_4}\big(-\phi_4h_3^\al c_\al p_3+h^\al_4c_\al p_4\big)=&2E.
\end{split}
\end{equation}
Arbitrary functions are restricted by inequality
\begin{equation}                                                  \label{inmdsk}
  \det g^{**}=\frac{(\phi_4)^2\big(h_3^2h_4^1-h_3^1h_4^2\big)^2}
  {(1-\phi_3\phi_4)^4}\ne0,\infty.
\end{equation}
Thus canonical separable metric of type $[2,0,2]_1$ is parameterized by six
arbitrary functions of single arguments satisfying requirement (\ref{inmdsk}).
Two nonzero functions, e.g., $h_3^2$ and $h_4^1$ can be set unity. Note that
$\det g^{**}$ is always positive. Therefore signature of the separable metric
of type $[2,0,2]_1$ can be only $(++--)$.

5) {\bf Class $[2,1,1]_1$.} Two commuting Killing vectors, one indecomposable
quadratic conservation law and one coisotropic coordinate. Coordinates and
parameters:
\begin{equation*}
  (x^\al,y^\mu,z^\vf)=(x^1,x^2,y^3:=y,z^4:=z),\qquad (c_\al,d_{ii},a_r)\mapsto
  (c_1,c_2,d_{33}:=d,a_4:=2E).
\end{equation*}
This case is described by theorem \ref{tnbdhg}. Matrix $B$ is parameterized
similar to Eq.~(\ref{iddkll}):
\begin{equation*}
  B=\begin{pmatrix} 1 & \phi_3 \\
  \displaystyle\frac{\phi_4}{h_4^\al c_\al} & \displaystyle\frac1{h_4^\al c_\al}
  \end{pmatrix}\qquad\Rightarrow\qquad
  B^{-1}=\frac1{1-\phi_3\phi_4}\begin{pmatrix} 1 & -\phi_3h_4^\al c_\al \\
  -\phi_4 & h_4^\al c_\al \end{pmatrix},
\end{equation*}
where $\phi_3(y)$, $\phi_4(z)$, and $h_4^\al(z)$, $\al=1,2$, are arbitrary
functions of single coordinates. It yields separable metric
\begin{equation}                                                  \label{iekdwa}
  g^{**}=\frac1{1-\phi_3\phi_4} \begin{pmatrix}
  \phi_4k_3^{\al\bt} & 0 & h_4^\al/2 \\ 0 & -\phi_4 & 0 \\ h_4^\bt/2 & 0 & 0
  \end{pmatrix},
\end{equation}
where functions $k_3^{\al\bt}(y)=k_3^{\bt\al}(y)$ are arbitrary. After
separation of the first two coordinates, $W'_\al=c_\al$, the Hamilton--Jacobi
equation takes the form
\begin{equation}                                                  \label{ifkwss}
  \frac1{1-\phi_3\phi_4}\big(\phi_4k_3^{\al\bt}c_\al c_\bt-\phi_4W^{\prime2}_3
  +h_4^\al c_\al W'_4\big)=2E.
\end{equation}
Variables are separated as
\begin{equation}                                                  \label{inbdgf}
\begin{split}
  W^{\prime2}_3=&d+2\phi_3E+k_3^{\al\bt}c_\al c_\bt,
\\
  W'_4=&\frac1{h^\al_4c_\al}(\phi_4 d+2E).
\end{split}
\end{equation}
Conservation laws are
\begin{equation}                                                  \label{indbdf}
\begin{split}
  \frac1{1-\phi_3\phi_4}\big(p_3^2-k_3^{\al\bt}c_\al c_\bt
  -\phi_3h_4^\al c_\al p_4\big)=&d,
\\
  \frac1{1-\phi_3\phi_4}\big[-\ph_4(p_3^2-k_3^{\al\bt}c_\al c_\bt)
  +h_4^\al c_\al p_4\big]=&2E.
\end{split}
\end{equation}
Determinant of metric (\ref{iekdwa}),
\begin{equation}                                                  \label{idhejt}
  \det g^{**}=\frac{\phi_4^2}{4(1-\phi_3\phi_4)^4}\left[
  k_3^{11}\big(h^2_4\big)^2+k_3^{22}\big(h^1_4\big)^2-2k_3^{12}h^1_4h^2_4
  \right]\ne0,\infty,
\end{equation}
restricts admissible form of arbitrary functions. Thus the canonical separable
metric (\ref{iekdwa}) is parameterized by seven arbitrary functions restricted
by Eq.~(\ref{idhejt}). In addition, they must admit real solutions for $W'_3$ in
Eq.~(\ref{inbdgf}). One nonzero component of $h_4^\al(z)$ can be transformed to
unity by suitable canonical transformation.

6) {\bf Class $[2,1,1]_2$.} Two commuting Killing vectors, one indecomposable
quadratic conservation law, one coisotropic coordinate. Coordinates and
parameters:
\begin{equation*}
  (x^\al,y^\mu,z^\vf)=(x^1,x^2,y^3:=y,z^4:=z),\qquad (c_\al,d_{ii},a_r)\mapsto
  (c_1,c_2,d_{33}:=2E,a_4=a).
\end{equation*}
This case is described by theorem \ref{tiswgk}. Matrix $B$ is the same as in the
previous case, but now the canonical separable metric is defined by the first
row of inverse matrix $B^{-1}$:
\begin{equation}                                                  \label{ibjvdf}
  g^{**}=\frac1{1-\phi_3\phi_4}\begin{pmatrix}
  \displaystyle-k_3^{\al\bt} & 0 & -\phi_3h_4^\al/2 \\
  0 & 1 & 0 \\ -\phi_3h^\bt_4/2 & 0 & 0 \end{pmatrix}.
\end{equation}
where functions $k_3^{\al\bt}(y)=k_3^{\bt\al}(y)$, $\al,\bt=1,2$, are arbitrary.
The Hamilton--Jacobi equation after separation of the first two coordinates
is
\begin{equation*}
  \frac1{1-\phi_3\phi_4}\big(-k_3^{\al\bt}c_\al c_\bt+W_3^{\prime2}
  -\phi_3h_4^\al c_\al W'_4\big)=2E.
\end{equation*}
Variables are separated as
\begin{equation}                                                  \label{inbdgi}
\begin{split}
  W^{\prime2}_3=&2E+\phi_3a+k_3^{\al\bt}c_\al c_\bt,
\\
  W'_4=&\frac1{h^\al_4c_\al}(2\phi_4 E+a).
\end{split}
\end{equation}
After substitution $W'_\al=c_\al$, conservation laws are
\begin{equation}                                                  \label{ikdbdf}
\begin{split}
  \frac1{1-\phi_3\phi_4}\big(p_3^2-k_3^{\al\bt}c_\al c_\bt
  -\phi_3h_4^\al c_\al p_4\big)=&2E,
\\
  \frac1{1-\phi_3\phi_4}\big[-\phi_4(p_3^2-k_3^{\al\bt}c_\al c_\bt)
  +h^\al_4c_\al p_4\big]=&a.
\end{split}
\end{equation}
Separable metric (\ref{ibjvdf}) is parameterized by seven arbitrary functions of
single coordinates which are restricted by inequality
\begin{equation*}
  \det g^{**}=\frac{\phi_3^2}{4(1-\phi_3\phi_4)^4}\left[
  k_3^{11}\big(h^2_4\big)^2+k_3^{22}\big(h^1_4\big)^2-2k_3^{12}h^1_4h^2_4
  \right]\ne0,\infty.
\end{equation*}
Also, Eqs.~(\ref{inbdgi}) must admit real solutions for $W'_3$. One of the
nonzero functions $h_4^\al$ can be set to unity by suitable canonical
transformation.

7) {\bf Class $[1,3,0]_2$.} One Killing vector and three indecomposable
quadratic conservation laws. Coordinates and parameters:
\begin{equation*}
\begin{split}
  (x^\al,y^\mu,z^\vf)&\mapsto(x^1:=x,y^2,y^3,y^4),
\\
  (c_\al,d_{ii},a_r)&\mapsto(c_1:=c,d_{22}:=d_2,d_{33}:=d_3,d_{44}:=2E).
\end{split}
\end{equation*}
This case corresponds to theorem \ref{theojk}. Matrix $b$ has the same form as
in three dimensions (\ref{idlkih})
\begin{equation}                                                  \label{idukih}
  b_{\mu\mu}{}^{ii}=
  \begin{pmatrix}1 & b_{23}(y^2) & b_{24}(y^2)\\ b_{32}(y^3) & 1 & b_{34}(y^3)\\
  b_{42}(y^4) & b_{43}(y^4) & 1\end{pmatrix},\qquad
  b_{ii}{}^{\mu\mu}=\frac1\vartriangle
  \begin{pmatrix} \vartriangle_{22} & \vartriangle_{32} & \vartriangle_{42} \\
  \vartriangle_{23} & \vartriangle_{33} & \vartriangle_{43} \\
  \vartriangle_{24} & \vartriangle_{34} & \vartriangle_{44} \end{pmatrix},
\end{equation}
where $\vartriangle:=\det b_{\mu\mu}{}^{ii}$ and symbols $\vartriangle_{\mu i}$
denote cofactors of elements $b_{\mu\mu}{}^{ii}$. Equations (\ref{anvbfg}) imply
that the separable metric must be diagonal
\begin{equation}                                                  \label{ighsii}
  g^{**}=\frac1{\vartriangle}
  \begin{pmatrix}-\vartriangle_{24}k_2-\vartriangle_{34}k_3
    -\vartriangle_{44}k_4 & 0 & 0 & 0 \\
    0 & \vartriangle_{24} & 0 & 0 \\ 0 & 0 & \vartriangle_{34} & 0 \\
    0 & 0 & 0 & \vartriangle_{44} \end{pmatrix},
\end{equation}
where $k_2(y^2)$, $k_3(y^3)$, and $k_4(y^4)$ are arbitrary functions of single
coordinates. The Hamilton--Jacobi equation after substitution $W'_1=c$ becomes
\begin{equation*}
  \frac1{\vartriangle}\big[(-\vartriangle_{24}k_2
  -\vartriangle_{34}k_3-\vartriangle_{44}k_4)c^2
  +\vartriangle_{24}W^{\prime2}_2+\vartriangle_{34}W^{\prime2}_3
  +\vartriangle_{44}W^{\prime2}_4\big]=2E.
\end{equation*}
Variables are separated as
\begin{equation}                                                  \label{ifkdht}
\begin{split}
  W^{\prime2}_2=d_2+b_{23}d_3+2b_{24}E+k_2c^2,
\\
  W^{\prime2}_3=b_{32}d_2+d_3+2b_{34}E+k_3c^2,
\\
  W^{\prime2}_4=b_{42}d_2+b_{43}d_3+2E+k_4c^2.
\end{split}
\end{equation}
All three conservation laws are quadratic
\begin{equation}                                                  \label{ihdjdx}
\begin{split}
  \frac1\vartriangle\big[\vartriangle_{22}(p_2^2-k_2c^2)
  +\vartriangle_{32}(p_3^2-k_3c^2)+\vartriangle_{42}(p_4^2-k_4c^2)\big]=&d_2,
\\
  \frac1\vartriangle\big[\vartriangle_{23}(p_2^2-k_2c^2)
  +\vartriangle_{33}(p_3^2-k_3c^2)+\vartriangle_{43}(p_4^2-k_4c^2)\big]=&d_3,
\\
  \frac1\vartriangle\big[\vartriangle_{24}(p_2^2-k_2c^2)
  +\vartriangle_{34}(p_3^2-k_3c^2)+\vartriangle_{44}(p_4^2-k_4c^2)\big]=&2E.
\end{split}
\end{equation}

So, canonical separable metric of class $[1,3,0]_2$ (\ref{ighsii}) is
parameterized by nine arbitrary functions $b_{23}(y^2)$, $b_{24}(y^2)$,
$b_{32}(y^3)$, $b_{34}(y^3)$, $b_{42}(y^4)$, $b_{43}(y^4)$, $k_2(y^2)$,
$k_3(y^3)$, and $k_4(y^4)$ of single arguments with restriction
$\det g^{**}\ne0$. Besides, equation (\ref{ifkdht}) must admit real solutions
for $W'_\mu$.

For example, the Friedmann metric \cite{Friedm22} belongs to this class.

8) {\bf Class $[1,2,1]_2$.} One Killing vector, two indecomposable quadratic
conservation laws, and one coisotropic coordinate. Coordinates and parameters:
\begin{equation*}
\begin{split}
  &(x^\al,y^\mu,z^\vf)\mapsto(x^1:=x,y^2,y^3,z^4:=z),
\\
  &(c_\al,d_{ij},a_r)\mapsto(c_1:=c,d_{22}:=d,d_{33}:=2E,a_4:=a).
\end{split}
\end{equation*}
This case is described by theorem \ref{tiswgk}. We parameterize matrix $B$ in
the following way
\begin{equation}                                                  \label{idshfd}
  B=\begin{pmatrix} 1 & \phi_{23} & \phi_{24} \\ \phi_{32} & 1 & \phi_{34} \\
  \phi_{42}/hc & \phi_{43}/hc & 1/hc \end{pmatrix}\qquad\Rightarrow\qquad
  B^{-1}=\frac1\vartriangle
  \begin{pmatrix} \vartriangle_{22} & \vartriangle_{32} & \vartriangle_{42} \\
  \vartriangle_{23} & \vartriangle_{33} & \vartriangle_{43} \\
  \vartriangle_{24} & \vartriangle_{34} & \vartriangle_{44} \end{pmatrix},
\end{equation}
where $\phi_{23}(y^2)$, $\phi_{24}(y^2)$, $\phi_{32}(y^3)$, $\phi_{34}(y^3)$,
$\phi_{42}(z)$, $\phi_{43}(z)$, and $h(z)$ are arbitrary functions of single
coordinates,
\begin{equation*}
  \vartriangle:=hc\det B=1+\phi_{23}\phi_{34}\phi_{42}+\phi_{24}\phi_{43}
  \phi_{32}-\phi_{23}\phi_{32}-\phi_{24}\phi_{42}-\phi_{34}\phi_{43},
\end{equation*}
and the notation is introduced for cofactors multiplied by $hc$:
\begin{equation*}
\begin{aligned}
  \vartriangle_{22}=& 1-\phi_{34}\phi_{43}, &
  \qquad\vartriangle_{32}=& \phi_{24}\phi_{43}-\phi_{23}, &
  \qquad\vartriangle_{42}=& (\phi_{23}\phi_{34}-\phi_{24})hc,
\\
  \vartriangle_{23}=& \phi_{34}\phi_{42}-\phi_{32}, &
  \vartriangle_{33}=& 1-\phi_{24}\phi_{42}, &
  \vartriangle_{43}=& (\phi_{32}\phi_{24}-\phi_{34})hc,
\\
  \vartriangle_{24}=& \phi_{43}\phi_{32}-\phi_{42}, &
  \vartriangle_{34}=& \phi_{42}\phi_{23}-\phi_{43}, &
  \vartriangle_{44}=& (1-\phi_{23}\phi_{32})hc.
\end{aligned}
\end{equation*}
Equations (\ref{anvbfg}) imply components of separable metric
\begin{equation*}
\begin{aligned}
  g^{11}=&-\frac1\vartriangle(\vartriangle_{23}k_2+\vartriangle_{33}k_3), &
  \qquad g^{22}=&\frac{\vartriangle_{23}}{\vartriangle},
\\
  g^{33}=&\frac{\vartriangle_{33}}{\vartriangle}, &
  2g^{14}=&\frac{\vartriangle_{43}}{\vartriangle c},
\end{aligned}
\end{equation*}
where $k_2(y^2)$ and $k_3(y^3)$ are arbitrary functions of single coordinates
which are defined by the second row of matrix $B^{-1}$. Consequently the
canonical separable metric is
\begin{equation}                                                  \label{idkwel}
  g^{**}=\frac1\vartriangle\begin{pmatrix}
  -\vartriangle_{23}k_2-\vartriangle_{33}k_3 & 0 & 0 & \vartriangle_{43}/(2c)\\
  0 & \vartriangle_{23} & 0 & 0 \\ 0 & 0 & \vartriangle_{33} & 0 \\
  \vartriangle_{43}/(2c) & 0 & 0 & 0 \end{pmatrix}.
\end{equation}
After separation of the first coordinate, the Hamilton--Jacobi equation becomes
\begin{equation*}
  \frac1\vartriangle\left[-(\vartriangle_{23}k_2+\vartriangle_{33}k_3)c^2
  +\vartriangle_{23}W^{\prime2}_2+\vartriangle_{33}W^{\prime2}_3
  +\vartriangle_{43}W'_4\right]=2E.
\end{equation*}
Variables are separated as
\begin{equation}                                                  \label{iddjkj}
\begin{split}
  W^{\prime2}_2=&d+2\phi_{23}E+\phi_{24}a+k_2,
\\
  W^{\prime2}_3=&\phi_{32}d+2E+\phi_{34}a+k_3,
\\
  W'_4=&\frac1{hc}(\phi_{42}d+2\phi_{43}E+a).
\end{split}
\end{equation}
The respective conservation laws are quadratic
\begin{equation*}
\begin{split}
  \frac1\vartriangle\big[\vartriangle_{22}(p_2^2-k_2c^2)
  +\vartriangle_{32}(p_3^2-k_3c^2)+\vartriangle_{42}p_4\big]=&d,
\\
  \frac1\vartriangle\big[\vartriangle_{23}(p_2^2-k_2c^2)
  +\vartriangle_{33}(p_3^2-k_3c^2)+\vartriangle_{43}p_4\big]=&2E,
\\
  \frac1\vartriangle\big[\vartriangle_{24}(p_2^2-k_2c^2)
  +\vartriangle_{34}(p_3^2-k_3c^2)+\vartriangle_{44}p_4\big]=&a,
\end{split}
\end{equation*}
Thus, the canonical separable metric (\ref{idkwel}) of class $[1,2,1]_2$ is
parameterized by nine arbitrary functions of single coordinates satisfying
inequality
\begin{equation*}
  \det g^{**}=-\frac{\vartriangle_{23}\vartriangle_{33}\vartriangle_{43}^2h^2}
  {4c^2\vartriangle^4}\ne0,\infty.
\end{equation*}
In addition, they are to provide real solutions of Eqs.~(\ref{iddjkj}) with
respect to $W'_\mu$ for fixed parameters $c$, $d$, $E$, and $a$.

9) {\bf Class $[1,2,1]_1$.} One Killing vector, two indecomposable quadratic
conservation laws, and one coisotropic coordinate. This case is described by
theorem \ref{tiswgk}. Coordinates and matrix $B$ are the same as in class
$[1,2,1]_2$. The difference reduces to the definition of parameters
\begin{equation*}
  (c_\al,d_{ij},a_r)\mapsto(c_1:=c,d_{22}:=d_2,d_{33}:=d_3,a_4:=2E).
\end{equation*}
Now the separable metric is defined by the last row of inverse matrix $B^{-1}$.
Equations (\ref{ahgdyu}) imply equalities
\begin{equation}                                                  \label{ifmrjh}
\begin{aligned}
  g^{11}=&-\frac1{\vartriangle}(\vartriangle_{24}k_2+\vartriangle_{34}k_3), &
  \qquad g^{22}=&\frac{\vartriangle_{24}}\vartriangle,
\\
  g^{33}=&\frac{\vartriangle_{34}}\vartriangle, & 2g^{14}=&
  \frac{\vartriangle_{44}}{\vartriangle c}.
\end{aligned}
\end{equation}
Therefore the canonical separable metric is
\begin{equation}                                                  \label{idkwej}
  g^{**}=\frac1\vartriangle\begin{pmatrix}
  -\vartriangle_{24}k_2-\vartriangle_{34}k_3 & 0 & 0 &
  \vartriangle_{44}/(2c) \\[8pt] 0 & \vartriangle_{24} & 0 &
  0 \\[4pt] 0 & 0 & \vartriangle_{34} & 0 \\[2pt]
  \vartriangle_{44}/(2c) & 0 & 0 & 0 \end{pmatrix}.
\end{equation}
Note that metric components, as in the previous case, do not depend on
parameters $c$. After separation of the first coordinate, the
Hamilton--Jacobi equation becomes
\begin{equation*}
  \frac1\vartriangle\left[-(\vartriangle_{24}k_2+\vartriangle_{34}k_3)c^2
  +\vartriangle_{24}W^{\prime2}_2+\vartriangle_{34}W^{\prime2}_3
  +\vartriangle_{44}W'_4\right]=2E.
\end{equation*}
Variables are separated in the following way
\begin{equation}                                                  \label{iddjkh}
\begin{split}
  W^{\prime2}_2=&d_2+\phi_{23}d_3+2\phi_{24}E+k_2,
\\
  W^{\prime2}_3=&\phi_{32}d_2+d_3+2\phi_{34}E+k_3,
\\
  W'_4=&\frac1{hc}\big(\phi_{42}d_2+\phi_{43}d_3+2E\big).
\end{split}
\end{equation}
The respective conservation laws are quadratic
\begin{equation*}
\begin{split}
  \frac1\vartriangle\big[\vartriangle_{22}(p_2^2-k_2c^2)
  +\vartriangle_{32}(p_3^2-k_3c^2)+\vartriangle_{42}p_4\big]=&d_2,
\\
  \frac1\vartriangle\big[\vartriangle_{23}(p_2^2-k_2c^2)
  +\vartriangle_{33}(p_3^2-k_3c^2)+\vartriangle_{43}p_4\big]=&d_3,
\\
  \frac1\vartriangle\big[\vartriangle_{24}(p_2^2-k_2c^2)
  +\vartriangle_{34}(p_3^2-k_3c^2)+\vartriangle_{44}p_4\big]=&2E.
\end{split}
\end{equation*}
Thus, canonical separable metric (\ref{idkwel}) of class $[1,2,1]_1$ is
parameterized by nine functions of single coordinates satisfying restriction
\begin{equation*}
  \det g^{**}=-\frac{\vartriangle_{24}\vartriangle_{34}\vartriangle_{44}^2}
  {4c^2\vartriangle^4}\ne0.
\end{equation*}
In addition Eqs.~(\ref{iddjkh}) must have real solutions $W'_\mu$ for fixed
parameters $c$, $d_2$, $d_3$, and $E$.

10) {\bf Class $[0,4,0]_2$.} Four indecomposable quadratic conservation laws
without coisotropic coordinates. Coordinates and parameters:
\begin{equation*}
\begin{split}
  &(x^\al,y^\mu,z^\vf)\mapsto(y^1,y^2,y^3,y^4),
\\
  &(c_\al,d_{ij},a_r)\mapsto(d_{11}:=d_1,d_{22}:=d_2,d_{33}:=d_3,d_{44}:=2E).
\end{split}
\end{equation*}
This case corresponds to theorem \ref{tgdjke}. Matrix $B$ is parameterized as
\begin{equation}                                                  \label{idlkih}
  b_{\mu\mu}{}^{ii}=
  \begin{pmatrix}1 & b_{12}(y^1) & b_{13}(y^1) & b_{14}(y^1) \\
  b_{21}(y^2) & 1 & b_{23}(y^2) & b_{24}(y^2) \\
  b_{31}(y^3) & b_{32}(y^3) & 1 & b_{34}(y^3) \\
  b_{41}(y^4) & b_{42}(y^4) & b_{43}(y^4) & 1\end{pmatrix},\qquad
  b_{ii}{}^{\mu\mu}=\frac1\vartriangle
  \begin{pmatrix}
  \vartriangle_{11} & \vartriangle_{21} & \vartriangle_{31} &\vartriangle_{41}\\
  \vartriangle_{12} & \vartriangle_{22} & \vartriangle_{32} &\vartriangle_{42}\\
  \vartriangle_{13} & \vartriangle_{23} & \vartriangle_{33} &\vartriangle_{43}\\
  \vartriangle_{14} & \vartriangle_{24} & \vartriangle_{34} &\vartriangle_{44}\\
  \end{pmatrix},
\end{equation}
where $\vartriangle:=\det b_{\mu\mu}{}^{ii}$ and symbols $\vartriangle_{\mu i}$
denote cofactors of elements $b_{\mu\mu}{}^{ii}$. This matrix implies the
diagonal separable metric (\ref{adfghf}) defined by the last row of inverse
matrix $b^{-1}$
\begin{equation}                                                  \label{avcfsg}
  g^{**}=\frac1\vartriangle \begin{pmatrix} \vartriangle_{14} & 0 & 0 & 0 \\
  0 & \vartriangle_{24} & 0 & 0 \\ 0 & 0 & \vartriangle_{34} & 0 \\
  0 & 0 & 0 & \vartriangle_{44} \end{pmatrix}.
\end{equation}
The respective Hamilton--Jacobi equation is
\begin{equation*}
  \frac1\vartriangle\big(\vartriangle_{14}W^{\prime2}_1
  +\vartriangle_{24}W^{\prime2}_2+\vartriangle_{34}W^{\prime2}_3
  +\vartriangle_{44}W^{\prime2}_4\big)=2E.
\end{equation*}
Variables are separated in the following way
\begin{equation}                                                  \label{abbcvf}
\begin{split}
  W^{\prime2}_1=&d_1+b_{12}d_2+b_{13}d_3+2b_{14}E,
\\
  W^{\prime2}_2=&b_{21}d_1+d_2+b_{23}d_3+2b_{24}E,
\\
  W^{\prime2}_3=&b_{31}d_1+b_{32}d_2+d_3+2b_{34}E.
\\
  W^{\prime2}_4=&b_{41}d_1+b_{42}d_2+b_{43}d_3+2E.
\end{split}
\end{equation}
All four conservation laws are quadratic:
\begin{equation}                                                  \label{abcvdl}
\begin{split}
  \frac1\vartriangle\big(\vartriangle_{11}p^2_1+\vartriangle_{21}p^2_2
  +\vartriangle_{31}p^2_3+\vartriangle_{41}p^2_4\big)=&d_1,
\\
  \frac1\vartriangle\big(\vartriangle_{12}p^2_1+\vartriangle_{22}p^2_2
  +\vartriangle_{32}p^2_3+\vartriangle_{42}p^2_4\big)=&d_2,
\\
  \frac1\vartriangle\big(\vartriangle_{13}p^2_1+\vartriangle_{23}p^2_2
  +\vartriangle_{33}p^2_3+\vartriangle_{43}p^2_4\big)=&d_3.
\\
  \frac1\vartriangle\big(\vartriangle_{14}p^2_1+\vartriangle_{24}p^2_2
  +\vartriangle_{34}p^2_3+\vartriangle_{44}p^2_4\big)=&2E.
\end{split}
\end{equation}
Thus canonical separable metric (\ref{avcfsg}) of class $[0,4,0]_2$ is
parameterized by twelve arbitrary functions of single coordinates entering
matrix $b$ (\ref{idlkih}) with condition $\vartriangle\ne0$. Moreover,
arbitrary functions must provide real solution $W'_\mu$ of Eqs.~(\ref{abbcvf})
for fixed values of parameters $d_{1,2,3}$, and $E$. In addition, every row of
matrix $b_{\mu\mu}{}^{ii}$ must have at least one nonconstant function, because
otherwise a Killing vector field appears.

The list of all separable metrics in four dimensions appeared in
\cite{Katana23B}. Seven of ten separable metrics were found in \cite{Obukho06}.
\section{Conclusion}
So the St\"ackel problem is completely solved for metrics of arbitrary
signature in arbitrary dimensions $n$. All separable metrics are divided into
equivalence classes $[\Sn,\Sm,n-\Sn-\Sm]_{1,2}$, characterized by the number
$\Sn$ of commuting Killing vector fields, number $\Sm$ of indecomposable
quadratic conservation laws, and the number $n-\Sn-\Sm$ of coisotropic
coordinates, indices 1 and 2 denoting the group of parameters (coisotropic or
quadratic, respectively) to which energy $E$ belongs. Within each equivalence
class separable metrics are related by canonical transformations and by
nondegenerate transformations of parameters which do not involve coordinates.
There is canonical (the most simple) separable metric in each class. In general,
its form is defined by theorems \ref{tiswgk} and \ref{tnbdhg}. These theorems
are constructive. As examples, we list all canonical separable metrics in two
(3 classes), three (6 classes), and four (10 classes) dimensions. These classes
include all separable metrics known up to now and contain several new classes
with coisotropic coordinates.

Different technique of finding separable metrics was used in
\cite{KalMil80,KalMil81}. We obtained the same block form of canonical
separable metrics, and types of metrics coincide. Moreover we gave more detailed
classification indicated by indices $1,2$ and succeeded in finding explicitly
all separating functions $W'_\al$, $\al=1,\dotsc,n$, and all conservation laws.

Separable metrics of indefinite signature open new perspectives in gravity
models. In general, equations for metric are very complicated, and contain many
unknown metric components. To simplify the problem, one usually supposes that
space-time has some symmetry, and this reduces the number of unknown variables
and equations. Almost all known exact solutions were obtained in this way, and,
as a rule, they yield separable Hamilton--Jacobi equations for geodesics. The
problem can be inverted. We may postulate the separability of the
Hamilton--Jacobi equation and substitute a separable metric in, for example,
Einstein equations. This also significantly reduces the number of independent
metric components and gives a hope to find new exact solutions. The important
feature of this approach is that vacuum Einstein equations reduce to the
system of nonlinear ordinary differential equations because separable metrics
contain only functions of single coordinates. This approach seems very
attractive because the class of separable metrics is much wider then the class
of symmetric metrics admitting complete separation of variables.

Most of known exact solutions in general relativity admit complete separation of
variables in the Hamilton--Jacobi equation. B.~Carter analyzed Einstein--Maxwell
equations with cosmological constant for separable metrics of class $[2,2,0]_2$
and showed that Schwarzschild, Reissner--Nordstr\"om, Kerr and many more
solutions belong to this class \cite{Carter68}. It is easily checked that Kasner
solution is of type $[3,1,0]_2$ and the Friedmann solution is of class
$[1,3,0]_2$. So there is an interesting problem to check to which class belongs
a known exact solution. It is nontrivial because solutions are often written in
a nonseparable coordinate systems. The opposite way of thought is to analyze
Einstein equations for all other separable metrics like it was done by B.~Carter
for metrics of type $[2,2,0]_2$. The advantage of this approach is that metric
depends on relatively small number of unknown functions of single variables, and
this gives a hope to find exact solutions of Einstein equations even without any
symmetry. Moreover, geodesic equations in such cases are Liouville integrable,
and the global structure of respective space-times can be, probably, analyzed
analytically.

\bibliographystyle{unsrt}
\end{document}